\documentclass{article}
\usepackage{iclr2022_conference,times}

\usepackage[utf8]{inputenc} 
\usepackage[T1]{fontenc}    
\usepackage{url}            
\usepackage{booktabs}       
\usepackage{caption}
\usepackage{subcaption}
\usepackage{amsfonts}       
\usepackage{nicefrac}       
\usepackage{microtype}      
\usepackage{bm}
\usepackage{diagbox}
\usepackage{color} 
\usepackage{soul}
\usepackage{amssymb}
\usepackage{amsmath}
\usepackage{amsfonts}
\usepackage{mathtools}
\usepackage{lineno}
\usepackage[pdftex]{graphicx}
\usepackage[toc,page,header]{appendix}
\usepackage{minitoc}
\doparttoc 
\faketableofcontents 
\usepackage[labelformat=simple]{subcaption}
\usepackage{wrapfig}
\newenvironment{smalleralign}[1][\small]
 {\par\nopagebreak\leavevmode\vspace*{-\baselineskip}%
  \skip0=\abovedisplayskip
  #1%
  \def\maketag@@@##1{\hbox{\m@th\normalfont\normalsize##1}}%
  \abovedisplayskip=\skip0
  \align}
 {\endalign\ignorespacesafterend}
\makeatother

\usepackage{natbib}
\usepackage{multicol, multirow}
\usepackage[ruled,vlined]{algorithm2e}

\newtheorem{theorem}{Theorem}
\newtheorem{definition}{Definition}
\newtheorem{proposition}{Proposition}
\newtheorem{corollary}{Corollary}
\newtheorem{lemma}{Lemma}

\newenvironment{proof}{\paragraph{Proof:}}{\hfill$\square$}

\SetKwInput{KwInput}{Input}                
\SetKwInput{KwOutput}{Output} 
\newtheorem{innercustomthm}{Theorem}
\newenvironment{customthm}[1]
  {\renewcommand\theinnercustomthm{#1}\innercustomthm}
  {\endinnercustomthm}
\providecommand{\customgenericname}{}
\newcommand{\newcustomtheorem}[2]{%
  \newenvironment{#1}[1]
  {%
  \renewcommand\customgenericname{#2}%
  \renewcommand\theinnercustomgeneric{##1}%
  \innercustomgeneric
  }
  {\endinnercustomgeneric}
}
\newcustomtheorem{customtheorem}{Theorem}

\usepackage{hyperref} 
\hypersetup{
    colorlinks=true,
    citecolor=blue, 
    linkcolor=blue,
    filecolor=magenta,      
    urlcolor=blue,
linktocpage}

\title{
LIGS: Learnable Intrinsic-Reward Generation Selection for Multi-Agent Learning 
}

\author{%
\textbf{David  Mguni$^{1}$\thanks{{{Correspondence to davidmguni@hotmail.com.}}},  Taher Jafferjee$^1$, Jianhong Wang$^{2}$, Oliver Slumbers$^{1,5}$, Nicolas Perez-Nieves$^{2}$, } \\ \textbf{Feifei Tong$^1$, Li Yang$^3$, Jiangcheng Zhu$^1$, Yaodong Yang$^4$,  Jun Wang$^5$} \\ 
  $^1$Huawei Technologies, $^2$Imperial College London,  $^3$Shanghaitech University, \\ $^4$Institute for AI, Peking University \& BIGAI, $^5$University College London }
  
\iclrfinalcopy 
\begin{document}

\maketitle

\maketitle
\begin{abstract}
Efficient exploration is important for reinforcement learners to achieve high rewards. In multi-agent systems, \textit{coordinated}  exploration and behaviour is critical for agents to jointly achieve optimal outcomes. In this paper, we introduce a new general framework for improving coordination and performance of multi-agent reinforcement learners (MARL). Our framework, named Learnable Intrinsic-Reward Generation Selection algorithm (LIGS) introduces
an adaptive learner, Generator that observes the agents and learns to construct
intrinsic rewards online that coordinate the agents’ joint exploration and joint behaviour.
Using a novel combination of MARL and switching
controls, LIGS determines the best states to learn to add intrinsic rewards which
leads to a highly efficient learning process. LIGS can subdivide complex tasks making them easier to solve and enables systems of MARL agents to quickly solve environments with sparse rewards. LIGS can seamlessly adopt existing MARL algorithms and, our theory shows that it ensures convergence to policies that deliver higher system performance. We demonstrate its superior performance in challenging tasks in Foraging and StarCraft II.
\end{abstract}

\section{Introduction}

Cooperative multi-agent reinforcement learning (MARL) has emerged as a powerful tool to enable autonomous agents to solve various tasks such as autonomous driving \citep{zhou2020smarts}, ride-sharing \citep{li2019efficient}, gaming AIs \citep{peng1703multiagent}, power networks \citep{wang2021multi,qiu2021multi} and swarm intelligence \citep{mguni2018decentralised,yang2017study}. In multi-agent systems (MAS), maximising system performance often requires agents to coordinate during exploration and learn coordinated joint actions. However, in many MAS, the reward signal provided by the environment is not sufficient to guide the agents towards coordinated behaviour \citep{matignon2012independent}. Consequently, relying on solely the individual rewards received by the agents may not lead to optimal outcomes \citep{mguni2019coordinating}.  This problem is exacerbated by the fact that MAS can have many stable points some of which lead to arbitrarily bad outcomes \citep{roughgarden2007introduction}.

As in single agent RL, in MARL inefficient exploration can dramatically decrease sample efficiency. In MAS, a major challenge is how to overcome sample inefficiency from poorly \textit{coordinated exploration}. Unlike single agent RL, in MARL, the collective of agents is typically required to coordinate its exploration to find their optimal joint policies\footnote{Unlike single agent RL, MARL exploration issues cannot be mitigated by adjusting
exploration rates or policy variances \citep{mahajan2019maven}.}.
A second issue is that in many MAS settings of interest, such as video games and physical tasks, rich informative signals of the agents' \textit{joint} performance are not readily available \citep{hosu2016playing}. For example, in StarCraft Micromanagement \citep{samvelyan2019starcraft}, the sparse reward alone (win, lose) gives insufficient information to guide agents toward their optimal joint policy. Consequently, MARL requires large numbers of samples producing a great need for MARL methods that can solve such problems efficiently.

To aid coordinated learning, algorithms such as QMIX \citep{rashid2018qmix}, MF-Q \citep{yang2018mean}, Q-DPP \citep{yang2020multi}, COMA \citep{foerster2018counterfactual} and SQDDPG \citep{wang2020shapley}, so-called centralised critic and decentralised execution (CT-DE) methods use a centralised critic whose role is to estimate the agents' expected returns. The critic makes use of all available information generated by the system, specifically the global state and the joint action \citep{peng2017multiagent}.  
To enable effective CT-DE, it is critical that the joint greedy action should be equivalent to the collection of individual greedy actions of agents, which is called the IGM (Individual-Global-Max) principle \citep{son2019qtran}.  
CT-DE methods are however, prone to convergence to suboptimal joint policies \citep{wang2020towards} and suffer from variance issues for gradient estimation  \citep{kuba2021settling}. Existing value factorisations, e.g. QMIX and VDN \citep{sunehag2017value} cannot ensure an exact guarantee of IGM consistency \citep{wang2020qplex}. Moreover, CT-DE methods such as QMIX require a monotonicity condition which is violated in scenarios where multiple agents must coordinate but are penalised
if only a subset of them do so (see Exp. 2, Sec. \ref{exp:foraging}).

To tackle these issues, in this paper we introduce a new MARL framework, LIGS that constructs intrinsic rewards online which guide MARL learners towards their optimal joint policy. LIGS involves an \textit{adaptive} intrinsic reward agent, 
the {\fontfamily{cmss}\selectfont Generator} that selects intrinsic rewards to add according to the history of visited states and the agents' joint actions. The {\fontfamily{cmss}\selectfont Generator} adaptively guides the agents' exploration and behaviour towards coordination and maximal joint performance. A pivotal feature of LIGS is the novel combination of RL and \textit{switching controls} \citep{mguni2018viscosity} which enables it to determine the best set of states to learn to add intrinsic rewards while disregarding less useful states. This enables the {\fontfamily{cmss}\selectfont Generator} to quickly learn how to set intrinsic rewards that guide the agents during their learning process.  
Moreover, the intrinsic rewards added by the {\fontfamily{cmss}\selectfont Generator} can significantly deviate from the environment rewards. This enables LIGS to both promote complex \textit{joint exploration} patterns and decompose difficult tasks. Despite this flexibility, special features within LIGS ensure the underlying optimal policies are preserved so that the agents learn to solve the task at hand.



Overall, LIGS has several advantages:  
\newline$\bullet$ LIGS has the freedom to introduce rewards that vastly deviate from the environment rewards. With this, LIGS promotes \textit{coordinated exploration} (i.e. visiting unplayed state-joint actions) among the agents enabling them to find joint policies that maximise the system rewards and generates intrinsic rewards to aid solving sparse reward MAS (see Experiment 1 in Sec. \ref{exp:foraging}).
\newline$\bullet$ LIGS selects which best states to add intrinsic rewards \textit{adaptively} in response to the agents' behaviour while the agents learn leading to an efficient learning process (see \textit{Investigations} in Sec. \ref{exp:foraging}).
\newline$\bullet$ LIGS's intrinsic rewards preserve the agents' optimal joint policy and ensure that the total \textit{environment} return is (weakly) increased (see Sec. \ref{sec:convergence}).

To enable the framework to perform successfully, we overcome several challenges: 
\textbf{i)} Firstly, constructing an intrinsic reward can change the underlying problem leading to the agents solving irrelevant tasks  \citep{mannion2017policy}. We resolve this by endowing the intrinsic reward function with special form which both allows a rich spread of intrinsic rewards while preserving the optimal policy. 
\textbf{ii)} Secondly, introducing intrinsic reward functions can \textit{worsen} the agents' performance \citep{devlin2011theoretical} and doing so \textit{while training} can lead to convergence issues.  We prove LIGS leads to better performing policies and that LIGS's learning process converges and preserves the MARL learners' convergence properties. 
\textbf{iii)} Lastly, adding an agent {\fontfamily{cmss}\selectfont Generator} with its own goal leads to a Markov game (MG) with $N+1$ agents \citep{fudenberg1991tirole}. 
Tractable methods for solving MGs are extremely rare with convergence only in special cases  \citep{yang2020overview}. Nevertheless, using a special set of features in LIGS's design, we prove LIGS converges to a solution in which it learns an intrinsic reward function that improves the agents' performance.

\section{Related Work}
%
%
%
%
%
%
\textbf{Reward shaping} \citep{harutyunyan2015expressing,mguni2021learning} is a technique which aims to alleviate the problem of sparse and uninformative rewards by supplementing the agent's reward with a prefixed term $F$. 
In \cite{ng1999policy} it was established that adding a \textit{shaping reward function} of the form $F(s_{t+1},s_{t})=\gamma\phi(s_{t+1})-\phi(s_t)$ preserves the optimal policy and in some cases can aid learning. 
%
RS has been extended to MAS \citep{devlin2011empirical,mannion2018reward,devlin2011theoretical, devlin2012dynamic, devlin2016plan,sadeghlou2014dynamic} in which it is used to promote convergence to efficient social welfare outcomes. 
Poor choices of $F$ in a MAS can slow the learning process and 
can induce convergence to poor system performance \citep{devlin2011theoretical}. In MARL, the question of which shaping function to use remains unaddressed.
Typically, RS algorithms rely on hand-crafted shaping reward functions that are constructed using domain knowledge, contrary to the goal of autonomous learning~\citep{devlin2011theoretical}. As we later describe LIGS, which successfully \textit{learns} an instrinsic reward function $F$, uses a similar form as PBRS however, $F$ is now augmented to include the actions of another RL agent to learn the intrinsic rewards online. In \cite{du2019liir} an approach towards learning intrinsic rewards is proposed in which a parameterised intrinsic reward is learned using a bilevel approach through a centralised critic. In \citet{WangZKG21}, a parameterised intrinsic reward is learned by a corpus, then the trained intrinsic reward is frozen on parameters and used to assist the training of a single-agent policy for generating the dialogues. Loosely related are single-agent methods \citep{zheng2018learning, dilokthanakul2019feature, kulkarni2016hierarchical, pathak2017curiosity} which, in general, introduce heuristic terms to generate intrinsic rewards. 


\textbf{Multi-agent exploration methods} seek to promote coordinated exploration among MARL learners. \cite{mahajan2019maven} proposed a hybridisation of value and policy-based methods that uses mutual information to learn a diverse set of behaviours between agents. Though this approach promotes coordinated exploration, it does not encourage exploration of novel states. Other approaches to promote exploration in MARL while assuming aspects of the environment are known in advance and agents can perform perfect communication between themselves \citep{viseras2016decentralized}. Similarly, to promote coordinated exploration in partially observable settings, \cite{pesce2020improving} proposed end-to-end learning of a communication protocol through a memory device. In general, exploration-based methods provide no performance guarantees nor do they ensure the optimal policy. Moreover, many employ heuristics that naively reward exploration to unvisited states without consideration of the environment reward. This can lead to spurious objectives being maximised.

Within these categories, closest to our work is the intrinsic reward approach in \cite{du2019liir}. There, the agents' policies and intrinsic rewards are learned with a bilevel approach. In contrast, LIGS performs these operations \textit{concurrently} leading to a fast and efficient procedure. A crucial point of distinction is that in LIGS, the intrinsic rewards are constructed by an RL agent ({\fontfamily{cmss}\selectfont Generator}) with its own reward function. Consequently, LIGS can generate complex patterns of intrinsic rewards, encourage \textit{joint exploration}. Additionally, LIGS learns intrinsic rewards only at relevant states, this confers high computational efficiency. Lastly, unlike exploration-based methods e.g., \cite{mahajan2019maven}, LIGS ensures preservation of the agents' joint optimal policy for the task.
%
\section{Preliminaries}
%
A fully cooperative MAS is modelled by a decentralised-Markov decision process (Dec-MDP) \citep{deng2021complexity}. A Dec-MDP is an augmented MDP involving a set of $N\geq 2$  agents denoted by $\mathcal{N}$ that independently decide actions to take which they do so simultaneously over many rounds. Formally, a dec-MDP is a tuple $\mathfrak{M}=\langle \mathcal{N},\mathcal{S},\left(\mathcal{A}_{i}\right)_{i\in\mathcal{N}},P,R,\gamma\rangle$ where $\mathcal{S}$ is the finite set of states, $\mathcal{A}_i$ is an action set for agent $i\in\mathcal{N}$ and $R:\mathcal{S}\times\boldsymbol{\mathcal{A}}\to\mathcal{P}(D)$ is the reward function that all agents jointly seek to maximise where $D$ is a compact subset of $\mathbb{R}$ and lastly, $P:\mathcal{S} \times \boldsymbol{\mathcal{A}} \times \mathcal{S} \rightarrow [0, 1]$ is the probability function describing the system dynamics where $\boldsymbol{\mathcal{A}}:=\times_{i=1}^N\mathcal{A}_i$.  Each agent $i\in\mathcal{N}$ uses a \textit{Markov policy}
$\pi_{i}: \mathcal{S} \times \mathcal{A}_i \rightarrow [0,1]$ to select its actions. At each time $t\in 0,1,\ldots,$ the system is in state $s_t\in\mathcal{S}$ and each agent $i\in\mathcal{N}$ takes an action $a^i_t\in\mathcal{A}_i$. The \textit{joint action}\ $\boldsymbol{a}_t=(a^1_t,\ldots, a^N_t)\in\boldsymbol{\mathcal{A}}$  produces an immediate reward $r_i\sim R(s_t,\boldsymbol{a}_t)$ for agent $i\in\mathcal{N}$ and influences the next-state transition which is chosen according to $P$.  
The goal of each agent $i$ is to maximise its expected returns measured by its value function $v^{\pi^i,\pi^{-i}}(s)=\mathbb{E}_{\pi^i,\pi^{-i}}\left[\sum_{t=0}^\infty \gamma^tR(s_t,\boldsymbol{a}_t)\right]$,
where $\Pi_i$ is a compact Markov policy space and $-i$ denotes the tuple of agents excluding agent $i$.
 
Intrinsic rewards can strongly induce more efficient learning (and can promote convergence to higher performing policies) \citep{devlin2011theoretical}. We tackle the problem of how to \textit{learn} intrinsic rewards produced by a function $F$ that leads to the agents learning policies that jointly maximise the system performance (through coordinated learning).
%
Determining this function is a significant challenge since poor choices can hinder learning and the concurrency of multiple learning processes presents potential convergence issues in a system already populated by multiple learners \citep{zinkevich2006cyclic}. Additionally, we require that the method preserves the optimal joint policy. 

\section{The LIGS Framework} 





To tackle the challenges described above, we introduce {\fontfamily{cmss}\selectfont Generator}  an \textit{adaptive} agent with its own objective that determines the best intrinsic rewards to give to the agents at each state.  Using observations of the joint actions played by the $N$ agents, the goal of the {\fontfamily{cmss}\selectfont Generator} is to construct intrinsic rewards to coordinate exploration and guide the agents towards learning joint policies that maximise their shared rewards. To do this, the {\fontfamily{cmss}\selectfont Generator} learns how to choose the values of an intrinsic reward function $F^{\boldsymbol{\theta}}$ at each state. 
Simultaneously, the $N$ agents perform actions to maximise their rewards using their individual policies.  The objective for each agent $i\in\{1,\ldots,N\}$ is given by: 
\begin{smalleralign}
v^{\pi^i,\pi^{-i},g}(s)=\mathbb{E}\left[\sum_{t=0}^\infty \gamma^t\left(R+F^{\boldsymbol{\theta}}\right)\Big|s_0=s\right], \nonumber
\end{smalleralign}
where 
$\boldsymbol{\theta}$ is determined by the {\fontfamily{cmss}\selectfont Generator} using the policy $g:\mathcal{S}\times\Theta\to[0,1]$ and $\Theta\subset\mathbb{R}^p$ is the  {\fontfamily{cmss}\selectfont Generator}'s action set. The intrinsic reward function is given by $
    F^{\boldsymbol{\theta}}(\cdot)\equiv \theta^c_t-\gamma^{-1}\theta^c_{t-1}$ 
%
where $\theta^c_t\sim g$  is the action chosen by the {\fontfamily{cmss}\selectfont Generator} and $\theta^c_t\equiv 0, \forall t<0$. 
$\Theta$ can be a set of integers $\{1,\ldots,K\}$).
Therefore, the {\fontfamily{cmss}\selectfont Generator} determines the output of $F^{\boldsymbol{\theta}}$ (which it does through its choice of $\theta^c$). With this, the {\fontfamily{cmss}\selectfont Generator} constructs intrinsic rewards that are tailored for the specific setting. 

 LIGS freely adopts any MARL algorithm for the $N$ agents (see Sec. \ref{sec:plug_n_play} in the Supp. Material).  The transition probability $P:\mathcal{S}\times\boldsymbol{\mathcal{A}}\times\mathcal{S}\to[0,1]$ takes the state and \textit{only} the actions of the $N$ agents as inputs.
 Note that unlike reward-shaping methods e.g. \citep{ng1999policy}, the function $F$ now contains action terms $\theta^c$ which are chosen by the {\fontfamily{cmss}\selectfont Generator} which enables the intrinsic reward function to be learned online. The presence of the action $\theta^c$ term may spoil the policy invariance result in \cite{ng1999policy}. We however prove a policy invariance result (Prop. \ref{invariance_prop}) analogous to that in \cite{ng1999policy} which shows LIGS preserves the optimal policy of $\mathfrak{M}$.
The {\fontfamily{cmss}\selectfont Generator} is an RL agent whose objective takes into account the history of states and $N$ agents' joint actions.
The {\fontfamily{cmss}\selectfont Generator}'s objective is:
\begin{smalleralign}
v^{\boldsymbol{\pi},g}_c(s)  = \mathbb{E}_{\boldsymbol{\pi},g}\left[ \sum_{t=0}^\infty \gamma^t\left(R^{\boldsymbol{\theta}}(s_t,\boldsymbol{a}_t) +L(s_t,\boldsymbol{a}_t)\right)\Big| s_0=s\right], \quad \forall s\in\mathcal{S}.
\end{smalleralign}
where $R^{\boldsymbol{\theta}}(s,\boldsymbol{a}):=R(s,\boldsymbol{a})+F^{\boldsymbol{\theta}}$. The objective encodes {\fontfamily{cmss}\selectfont Generator}'s agenda, namely to maximise the agents' expected return. Therefore, using its intrinsic rewards, the  {\fontfamily{cmss}\selectfont Generator} seeks to guide the set of agents toward optimal joint trajectories (potentially away from suboptimal trajectories, c.f. Experiment 2) and enables the agents to learn faster (c.f. StarCraft experiments in Sec. \ref{Section:Experiments}). Lastly, $L:\mathcal{S}\times\boldsymbol{\mathcal{A}}\to\mathbb{R}$ rewards {\fontfamily{cmss}\selectfont Generator} when the agents jointly visit novel state-joint-action tuples and tends to $0$ as the tuples are revisited. We later prove that with this objective, the {\fontfamily{cmss}\selectfont Generator}'s optimal policy (for constructing the intrinsic rewards) maximises the expected (extrinsic) return (Prop. \ref{preservation_lemma}).    

Since the {\fontfamily{cmss}\selectfont Generator} has its own (distinct) objective, the resulting setup is an MG, $\mathcal{G}=\langle \mathcal{N}\times\{c\},\mathcal{S},(\mathcal{A}_i)_{i\in\mathcal{N}},\Theta,P,R^{\boldsymbol{\theta}},R_c,\gamma\rangle$ where the new elements are $\{c\}$, the  {\fontfamily{cmss}\selectfont Generator} agent, $R^{\boldsymbol{\theta}}:=R+F^{\boldsymbol{\theta}}$, the new team reward function which contains the intrinsic reward $F^{\boldsymbol{\theta}}$,  $R_c:\mathcal{S}\times\boldsymbol{\mathcal{A}}\times\Theta\to\mathbb{R}$, the one-step reward for the {\fontfamily{cmss}\selectfont Generator} (we give the details of this later). 

\textbf{Switching Control Mechanism}

So far the {\fontfamily{cmss}\selectfont Generator}'s problem involves learning to construct intrinsic rewards at \textit{every} state which can be computationally expensive. We now introduce an important feature which allows LIGS to learn the best intrinsic reward only in a subset of states in which intrinsic rewards are most useful. This is in contrast to the problem tackled by the $N$ agents who must compute their optimal actions at all states. 
%
%
%
%
%
%
%
%
%
%
%
%
 To achieve this, we now replace the {\fontfamily{cmss}\selectfont Generator}'s policy space with a form of policies known as \textit{switching controls}. These policies enable {\fontfamily{cmss}\selectfont Generator} to decide at which states to learn the value of intrinsic rewards. This enables the {\fontfamily{cmss}\selectfont Generator} to learn quickly both where to add intrinsic rewards and the magnitudes that improve performance since the {\fontfamily{cmss}\selectfont Generator}'s magnitude optimisations are performed only at a subset of states. Crucially, with this the {\fontfamily{cmss}\selectfont Generator} can learn its policy rapidly enabling it to guide the agents toward coordination and higher performing policies while they train.  

At each state, the {\fontfamily{cmss}\selectfont Generator} first makes a \textit{binary decision} to decide to \textit{switch on} its $F$ for agent $i\in\mathcal{N}$ using a switch $I_t$ which takes values in $\{0,1\}$. 
Crucially, now the {\fontfamily{cmss}\selectfont Generator} is tasked with learning how to construct the $N$ agents' intrinsic rewards \emph{only} at states that are important for guiding the agents to their joint optimal policy. Both the decision to activate the function $F$ and its magnitudes is determined by the {\fontfamily{cmss}\selectfont Generator}. With this, the agent $i\in\mathcal{N}$ objective becomes: 
\begin{smalleralign}
v^{\boldsymbol{\pi},g}(s_0,I_0)=\mathbb{E}\left[\sum_{t=0}^\infty \gamma^t\left\{R+F^{\boldsymbol{\theta}}\cdot I_t\right\}\right],\; \forall (s_0,I_0)\in\mathcal{S}\times\{0,1\},
\end{smalleralign} 
where $I_{\tau_{k+1}}=1-I_{\tau_{k}}$,  which is the switch for $F$ which is $0$ or $1$ and $\{\tau_k\}_{k> 0}$ are times that a switch takes place\footnote{More precisely, $\{\tau_k\}_{k\geq 0}$ are \textit{stopping times}
\citep{oksendal2003stochastic}.} so for example if the switch is first turned on at the state $s_5$ then turned off at $s_7$, then $\tau_1=5$ and $\tau_2=7$ (we will shortly describe these in more detail). 
At any state, the decision to turn on $I$ is decided by a (categorical) policy $\mathfrak{g}_c:\mathcal{S} \to \{0,1\}$ which acts according to {\fontfamily{cmss}\selectfont Generator}'s objective. In particular, first, the {\fontfamily{cmss}\selectfont Generator} makes an observation of the state $s_k\in\mathcal{S}$ and the joint action $\boldsymbol{a}_k$ and using $\mathfrak{g}_c$, the {\fontfamily{cmss}\selectfont Generator} decides whether or not to activate the policy $g$ to provide an intrinsic reward whose value is determined by $\theta^c_k\sim g$. With this it can be seen the sequence of times $\{\tau_k\}$ is $\tau_k=\inf\{t>\tau_{k-1}|s_t\in\mathcal{S},\mathfrak{g}_c(s_t)=1\}$ so the switching times.
$\{\tau_k\}$ \textit{are \textbf{rules} that depend on the state.} Therefore, by learning an optimal $\mathfrak{g}_c$, the {\fontfamily{cmss}\selectfont Generator} learns the useful states to switch on $F$. 

To induce the {\fontfamily{cmss}\selectfont Generator} to selectively choose when to switch on the additional rewards, each switch activation  incurs a fixed cost for the {\fontfamily{cmss}\selectfont Generator}. In this case, the objective for the {\fontfamily{cmss}\selectfont Generator} is: 
\begin{smalleralign}
v^{\boldsymbol{\pi},g}_c(s_0,I_0)  = \mathbb{E}_{\boldsymbol{\pi},g}\left[ \sum_{t=0}^\infty \gamma^t\left(R^{\boldsymbol{\theta}}(s_t,\boldsymbol{a}_t) -\sum_{k\geq 1} \delta^t_{\tau_{2k-1}}
+L(s_t,\boldsymbol{a}_t)\right)\right], \label{generator_objective}
\end{smalleralign} 
where the 
Kronecker-delta function $\delta^t_{\tau_{2k-1}}$ which is $1$ whenever $t={\tau_{2k-1}}$ and $0$ otherwise imposes a cost for each switch activation. The cost has two effects: first, it reduces the computational complexity of the {\fontfamily{cmss}\selectfont Generator}'s problem since the {\fontfamily{cmss}\selectfont Generator} now determines \textit{subregions} of $\mathcal{S}$ it should learn the values of $F$. Second, it ensures the \textit{information-gain} from encouraging the agents to explore state-action tuples is sufficiently high to merit activating a stream of intrinsic rewards. 
%
%
%
%
%
%
%
%
We set $\tau_0\equiv 0$,  $\theta_{\tau_k}\equiv 0,\forall k\in\mathbb{N}$ ($\theta_{\tau_k+1},\ldots, \theta_{\tau_{k+1}-1}$ remain non-zero), $\theta^c_k\equiv 0\;\; \forall k\leq 0$ 
%
and denote by $I(t)\equiv I_t$.

\begin{wrapfigure}{r}{0.45\textwidth}
  \begin{center}
    \includegraphics[width=0.4\textwidth]{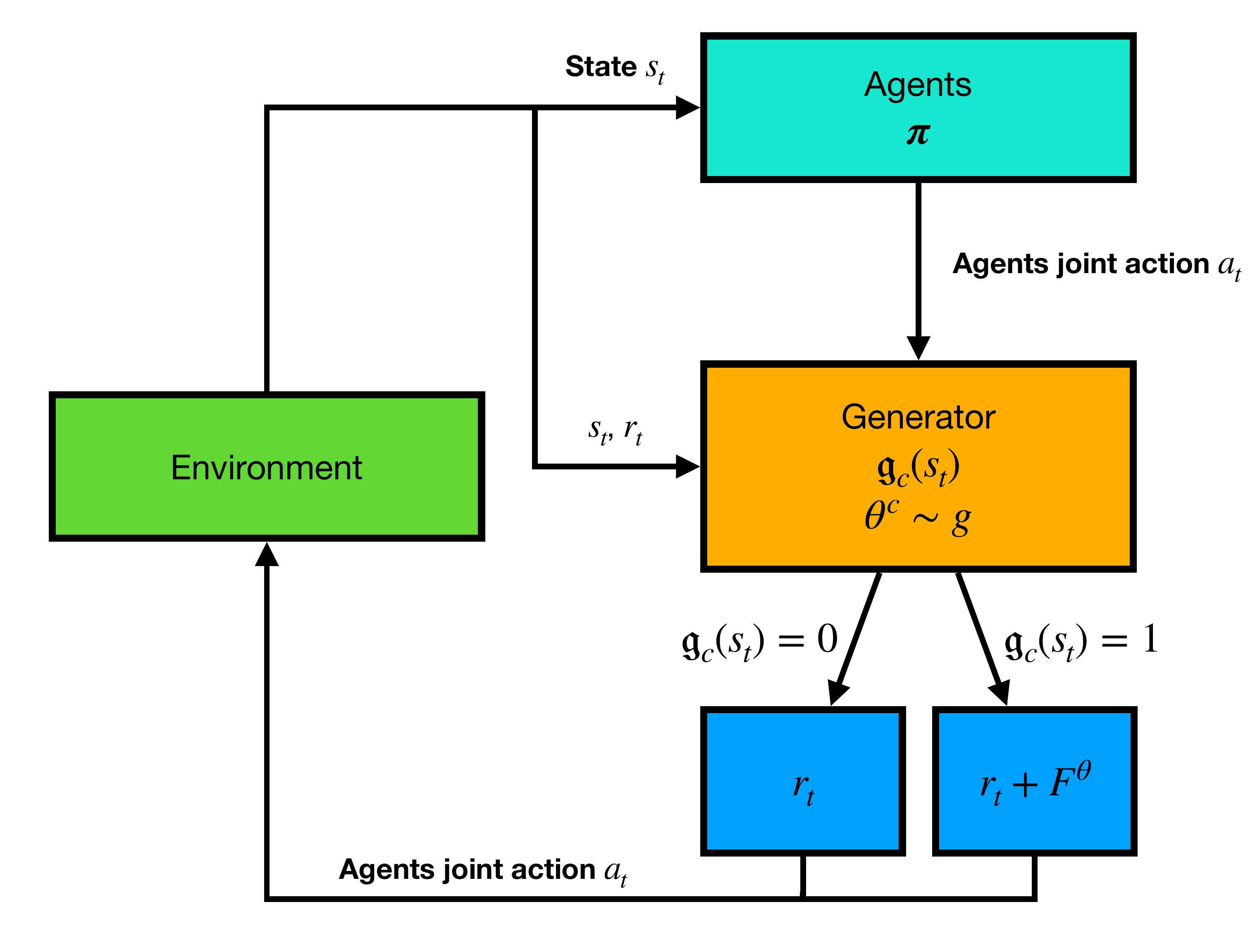}
  \end{center}
  \caption{Schematic of the LIGS framework.}
\end{wrapfigure}

\textbf{Discussion on Computational Aspect}\newline The switching controls mechanism results in a framework in which the problem facing the {\fontfamily{cmss}\selectfont Generator} has a markedly reduced decision space in comparison to the agent's problem  (though the agents share the same experiences). Crucially, the {\fontfamily{cmss}\selectfont Generator} must compute optimal intrinsic rewards at only a subset of states which are chosen by $\mathfrak{g}_c$. Moreover, the decision space for the switching policy $\mathfrak{g}_c$ is $\mathcal{S}\times\{0,1\}$ i.e at each state it makes a binary decision. Consequently, the learning process for $\mathfrak{g}_c$ is much quicker than the agents' policies which must optimise over the decision space $|\mathcal{S}||\mathcal{A}|$ (choosing an action at every state). This results in the {\fontfamily{cmss}\selectfont Generator} rapidly learning its optimal policies (relative to the agent) in turn, enabling the {\fontfamily{cmss}\selectfont Generator} to guide the agents towards its optimal policy during its learning phase. Also, in our experiments, we chose the size of the action set for the {\fontfamily{cmss}\selectfont Generator}, $\Theta$ to be a singleton resulting in a decision space of size $|\mathcal{S}|\times\{0,1\}$ for the entire problem facing the {\fontfamily{cmss}\selectfont Generator}. We later show that this choice leads to improved performance while removing the free parameter of the dimensionality of the {\fontfamily{cmss}\selectfont Generator}'s action set. 

\noindent\textbf{{Summary of Events}}

At a time $t\in 0,1,\ldots$\\
\begin{tabular}{l|m{0.9\linewidth}}
 $1.$\; & The $N$ agents makes an observation of the state $s_t\in\mathcal{S}$.\\
    
    $2.$\; & The $N$ agents perform a joint action $\boldsymbol{a}_t=(a^1_t,\ldots,a^N_t)$ sampled from $\boldsymbol{\pi}=(\pi^1,\ldots,\pi^N)$.\\
    
    $3.$\; & The {\fontfamily{cmss}\selectfont Generator} makes an observation of $s_t$ and $\boldsymbol{a}_t$ and draws samples from its polices $(\mathfrak{g}_c,g)$.\\
    
    $4.$\;
& If $\mathfrak{g}_c(s_t)=0$: \\
    
    & \textcolor{white}{X}$\circ$ Each agent $i\in\mathcal{N}$ receives a reward $r_i\sim R(s_t,\boldsymbol{a}_t)$ and the system transitions to the next state $s_{t+1}$ and steps 1 - 3 are repeated.\\
        
    $5.$\; & If $\mathfrak{g}_c(s_t)=1$:
    \\
        & \textcolor{white}{X}$\circ$
    $F^{\boldsymbol{\theta}}$ is computed using $s_t$ and the {\fontfamily{cmss}\selectfont Generator} action $\theta^c\sim g$. 
    \\& \textcolor{white}{X}$\circ$ Each agent $i\in\mathcal{N}$ receives a reward $r_i+F^{\boldsymbol{\theta}}$ and the system transitions to $s_{t+1}$.\\
    
    $6.$\; & At time $t+1$ if the intrinsic reward terminates then steps 1 - 3 are repeated or if the intrinsic reward has not terminated then step 5 is repeated.\\
    
\end{tabular}

\subsection{The Learning Procedure}\label{sec:learning_proc}

In Sec. \ref{sec:convergence}, we provide the convergence properties of the algorithm, 
and give the full code of the algorithm in Sec. \ref{sec:algorithm} of the Appendix.  The algorithm consists of the following procedures: the {\fontfamily{cmss}\selectfont Generator} updates its policy that determines the values $\theta$ at each state
and the states to perform a switch
 while the agents $\{1,\ldots,N\}$ learn their individual policies $\{\pi_1,\ldots,\pi_N\}$.
In our implementation, we used proximal policy optimization (PPO) \citep{schulman2017Proximal} as the learning algorithm for both the {\fontfamily{cmss}\selectfont Generator}'s intervention policy $\mathfrak{g}_c$ and {\fontfamily{cmss}\selectfont Generator}'s policy $g$. For the $N$ agents we used MAPPO \citep{yu2021surprising}.
for the {\fontfamily{cmss}\selectfont Generator} $L$ term we use\footnote{This is similar to random network distillation \citep{burda2018exploration} however the input is over the space $\mathcal{A}\times\mathcal{S}$.} $L(s_{t}, \boldsymbol{a}_t):=\|\hat{h} - \mathit{h}\|_{2}^{2}$  where $\hat{h}$ is a random initialised network which is the target network which is fixed and $\mathit{h}$ is the \textit{prediction function} that is consecutively updated during training. We constructed $F^{\boldsymbol{\theta}}$ using a fixed neural network $f:\mathbb{R}^d \mapsto \mathbb{R}^m$ and a one-hot encoding of the action of the {\fontfamily{cmss}\selectfont Generator}. Specifically, 
$i(\theta^c_t)$ is a one-hot encoding of the action $\theta^c_t$ picked by the {\fontfamily{cmss}\selectfont Generator}. Thus, $F^{(\theta^c_t, \theta^c_{t-1})} = i(\theta^c_t)  - \gamma^{-1} i(\theta^c_{t-1})$. The action set of the {\fontfamily{cmss}\selectfont Generator} is $\Theta \equiv \{1\}$ where 
$g$ is an MLP $g: \mathbb{R}^d    \mapsto \mathbb{R}^m$. Extra details are in Sec. \ref{sec:algorithm}.

\section{Convergence and Optimality of LIGS} \label{sec:convergence}

We now show that LIGS converges and that the solution ensures a higher performing agent policies.
The addition of the {\fontfamily{cmss}\selectfont Generator}'s RL process which modifies $N$ agents' rewards during learning can produce convergence issues \citep{zinkevich2006cyclic}.  Also to ensure the framework is useful, we must verify that the solution of $\mathcal{G}$ corresponds to solving the MDP, $\mathfrak{M}$. 
To resolve these issues, we first study the stable point solutions of $\mathcal{G}$.  Unlike MDPs, the existence of a solution in Markov policies is not guaranteed for MGs \citep{blackwell1968big} and is rarely computable (except for special cases such as \textit{team} and \textit{zero-sum} MGs \citep{shoham2008multiagent}).
MGs also often have multiple stable points that can be inefficient \citep{mguni2019coordinating}; in $\mathcal{G}$ such stable points would lead to a poor performing agent joint policy.
We resolve these challenges with the following scheme: 

\noindent\textbf{[I]} LIGS preserves the optimal solution of $\mathfrak{M}$.\newline 
\noindent\textbf{[II]}  The MG induced by LIGS has a stable point which is the convergence point of MARL.\newline 
\noindent\textbf{[III]} LIGS yields a team payoff that is (weakly) greater than that from solving $\mathfrak{M}$ directly.\newline 
\noindent\textbf{[IV]}  LIGS  converges to the solution with a linear function approximators.
In what follows, we denote by $\boldsymbol{\Pi}:=\times_{i\in\mathcal{N}}\Pi_i$. The results are built under Assumptions 1 - 7 (Sec. \ref{sec:notation_appendix} of the Appendix) which are standard in RL and stochastic approximation theory.
We now prove the result \textbf{[I]} which shows the solution to $\mathfrak{M}$ is preserved under the influence of LIGS: 

\begin{proposition}
\label{preservation_lemma} The following statements hold:
\newline i) $
\underset{\boldsymbol{\pi}\in\boldsymbol{\Pi}}{\max}\; v^{\boldsymbol{\pi},g}(s,\cdot)=\underset{\boldsymbol{\pi}\in\boldsymbol{\Pi}}{\max}\; v^{\boldsymbol{\pi}}(s),\;\forall s\in\mathcal{S}, \forall i \in\mathcal{N}, \forall g$ where $v^{\boldsymbol{\pi}}(s)=\mathbb{E}_{\boldsymbol{\pi}}\left[\sum_{t=0}^\infty \gamma^tR\right]$.
\newline ii) The {\fontfamily{cmss}\selectfont Generator}'s optimal policy maximises $v^{\boldsymbol{\pi}}(s)=\mathbb{E}\left[\sum_{t=0}^\infty \gamma^tR(s_t,\boldsymbol{a}_t)\right]$ for any
$s\in\mathcal{S}$. 

\end{proposition}
Result (i) says that the agents' problem is preserved under the {\fontfamily{cmss}\selectfont Generator}'s influence. Moreover the agents' (expected) total return is that from the environment (extrinsic rewards). Result (ii) establishes that the {\fontfamily{cmss}\selectfont Generator}'s optimal policy induces it to maximise the agents' joint (extrinsic) total return. The result is proven by a careful adaptation of the policy invariance result in \cite{ng1999policy} to our MARL switching control setting where the intrinsic-reward is not added at all states. Building on Prop. \ref{preservation_lemma}, we deduce the following result:
\begin{corollary}\label{invariance_prop}
LIGS preserves the dec-MDP played by the agents. In particular, let $(\boldsymbol{\hat{\pi}},\hat{g})$ be a stable point policy profile\footnote{By stable point profile we mean 
a Markov perfect equilibrium (MPE) \citep{fudenberg1991tirole}.} of the MG induced by LIGS, $\mathcal{G}$ then $\boldsymbol{\hat{\pi}}$ is a solution to the dec-MDP, $\mathfrak{M}$. 
\end{corollary}
Therefore, the introduction of the {\fontfamily{cmss}\selectfont Generator} does not alter the fundamentals of the problem. 
Our next task is to prove the existence of a stable point of the MG induced by LIGS and show it is a limit point of a sequence of Bellman operations. 
To do this we prove that a stable solution of $\mathcal{G}$ exists and that $\mathcal{G}$ has a special property that permits its stable point to be found using dynamic programming. 
%
%
%
The following result establishes that the solution of the MG $\mathcal{G}$, can be computed using RL methods:

\begin{theorem}\label{theorem:existence_2}
Given a function $V:\mathcal{S}\times\boldsymbol{\mathcal{A}}\to\mathbb{R}$,  $\mathcal{G}$ has a stable point given by $\underset{k\to\infty}{\lim}T^kV^{\boldsymbol{\pi},g}=\underset{{\boldsymbol{\hat{\pi}}}\in\boldsymbol{\Pi}}{\sup}V^{\boldsymbol{\pi},\hat{g}}=V^{\boldsymbol{\pi^\star},g^\star}$ where $(\boldsymbol{\pi^\star},g)$ is a stable solution of $\mathcal{G}$ and $T$ is the Bellman operator (c.f. \eqref{bellman_op}).
\end{theorem}
Theorem \ref{theorem:existence_2} proves that the MG $\mathcal{G}$ (which is the game that is induced when {\fontfamily{cmss}\selectfont Generator} plays with the $N$ agents) has a stable point which is the limit of a dynamic programming method. In particular, it proves the that the stable point of $\mathcal{G}$ is the limit point of the sequence $T^1V,T^2V,\ldots,$. Crucially, (by Corollary \ref{invariance_prop}) the limit point corresponds to the solution of the dec-MDP $\mathcal{M}$.  Theorem \ref{theorem:existence_2} is proven by firstly proving that $\mathcal{G}$ has a dual representation as an MDP whose solution corresponds to the stable point of the MG.  
Theorem \ref{theorem:existence_2} enables us to tackle the problem of finding the solution to $\mathcal{G}$ using distributed learning methods i.e. MARL to solve $\mathcal{G}$.   
%
Moreover, Prop. \ref{invariance_prop} indicates by computing the stable point of $\mathcal{G}$ leads to a solution of $\mathfrak{M}$. These results combined prove \textbf{[II]}. Our next result characterises the {\fontfamily{cmss}\selectfont Generator} policy $\mathfrak{g}_c$ and the optimal times to activate $F$. The result yields a key aspect of our algorithm for executing the  {\fontfamily{cmss}\selectfont Generator} activations of intrinsic rewards:
\begin{proposition}\label{prop:switching_times}
The policy $\mathfrak{g}_c$ is given by: $\mathfrak{g}_c(s_t,I_t)=H(\mathcal{M}^{\boldsymbol{\pi},g}V^{\boldsymbol{\pi},g}- V^{\boldsymbol{\pi},g})(s_t,I_t),\;\;\forall (s_t,I_t)\in\mathcal{S}\times\{0,1\}$, where $V^{\boldsymbol{\pi},g}$ is the solution in Theorem \ref{theorem:existence_2}, $\mathcal{M}$ is the  {\fontfamily{cmss}\selectfont Generator}'s intervention operator (c.f. \eqref{intervention_op}) and $H$ is the Heaviside function, moreover 
$\tau_k=\inf\{\tau>\tau_{k-1}|\mathcal{M}^{\boldsymbol{\pi},g}V^{\boldsymbol{\pi},g}= V^{\boldsymbol{\pi},g}\}$.

\end{proposition}



In general, introducing intrinsic rewards or shaping rewards may undermine learning and worsen overall performance. We now prove that the LIGS framework introduces an intrinsic reward which yields better performance for the $N$ agents as compared to solving $\mathfrak{M}$ directly (\textbf{[III]}).  

\begin{theorem}\label{NE_improve_prop}
Each agent's expected return $v^{\boldsymbol{\pi},g}$ whilst playing $\mathcal{G}$ is (weakly) higher than the expected return for $\mathfrak{M}$ (without the {\fontfamily{cmss}\selectfont Generator}) i.e. $v^{\boldsymbol{\pi},g}(s,\cdot)\geq v^{\boldsymbol{\pi}}(s),\;\forall s \in\mathcal{S},\;\forall i \in\mathcal{N}$. 
\end{theorem}

Theorem \ref{NE_improve_prop} shows that the {\fontfamily{cmss}\selectfont Generator}'s influence leads to an improvement in the system performance. Note that by Prop. \ref{preservation_lemma}, Theorem \ref{NE_improve_prop} compares the environment (extrinsic) rewards accrued by the agents so that the presence of the {\fontfamily{cmss}\selectfont Generator} increases the total expected environment rewards.
We complete our analysis by extending Theorem \ref{theorem:existence_2} to capture (linear) function approximators which proves \textbf{[IV]}. We first define a \textit{projection} $\Pi$ by: $
\Pi \Lambda:=\underset{\bar{\Lambda}\in\{\Phi r|r\in\mathbb{R}^p\}}{\arg\min}\left\|\bar{\Lambda}-\Lambda\right\|$ for any function $\Lambda$.
%
%
\begin{theorem}\label{primal_convergence_theorem}
LIGS converges to the stable point of  $\mathcal{G}$, 
moreover, given a set of linearly independent basis functions $\Phi=\{\phi_1,\ldots,\phi_p\}$ with $\phi_k\in L_2,\forall k$. LIGS converges to a limit point $r^\star\in\mathbb{R}^p$ which is the unique solution to  $\Pi \mathfrak{F} (\Phi r^\star)=\Phi r^\star$ where
    $\mathfrak{F}\Lambda:=\hat{R}+\gamma P \max\{\mathcal{M}\Lambda,\Lambda\}$ . Moreover, $r^\star$ satisfies: $
    \left\|\Phi r^\star - Q^\star\right\|\leq (1-\gamma^2)^{-1/2}\left\|\Pi Q^\star-Q^\star\right\|$.
\end{theorem}
The theorem establishes the convergence of LIGS to a stable point (of $\mathcal{G}$) with the use of linear function approximators. The second statement bounds the proximity of the convergence point by the smallest approximation error that can be achieved given the choice of basis functions. 

\section{Experiments}\label{Section:Experiments}
We performed a series of experiments on the Level-based Foraging environment \citep{papoudakis2020comparative} to test if LIGS: \textbf{1.} Efficiently promotes joint exploration \textbf{2.} Optimises convergence points by inducing coordination. \textbf{3.} Handles sparse reward environments. In all tasks, we compared the performance of LIGS against MAPPO \citep{yu2021surprising}, QMIX \citep{rashid2018qmix}; intrinsic reward MARL algorithms LIIR \citep{du2019liir}, LICA \citep{zhou2020learning}, and a leading MARL exploration algorithm MAVEN \citep{mahajan2019maven}.
We then compared LIGS against these baselines in StarCraft Micromanagement II (SMAC) \citep{samvelyan2019starcraft}.
Lastly, we ran a detailed suite of ablation studies (see Appendix) in which we demonstrated LIGS' flexibility to accommodate i) different MARL learners, ii) different $L$ bonus terms for the {\fontfamily{cmss}\selectfont Generator} objective. We also demonstrated the necessity of the switching control component in LIGS and LIGS' improved use of exploration bonuses.

\subsection{Cooperative Foraging Tasks}\label{exp:foraging}

\begin{figure}[b]
\vspace{-5 mm}
\centering
\hspace{-7 mm}\includegraphics[width=3cm, height=2.5cm]{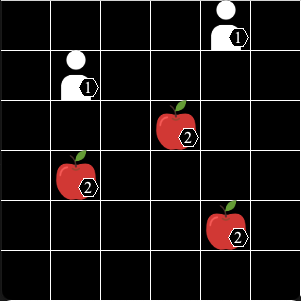}\hspace{12 mm}
\includegraphics[width=3cm, height=2.5cm]{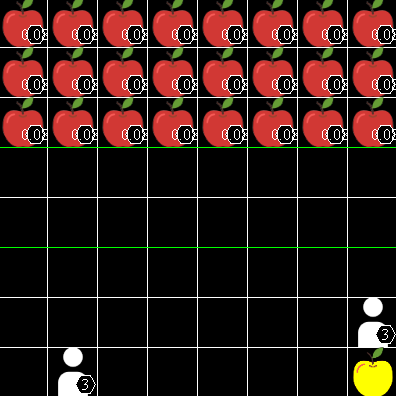}\hspace{9 mm}
\includegraphics[width=3cm, height=2.5cm]{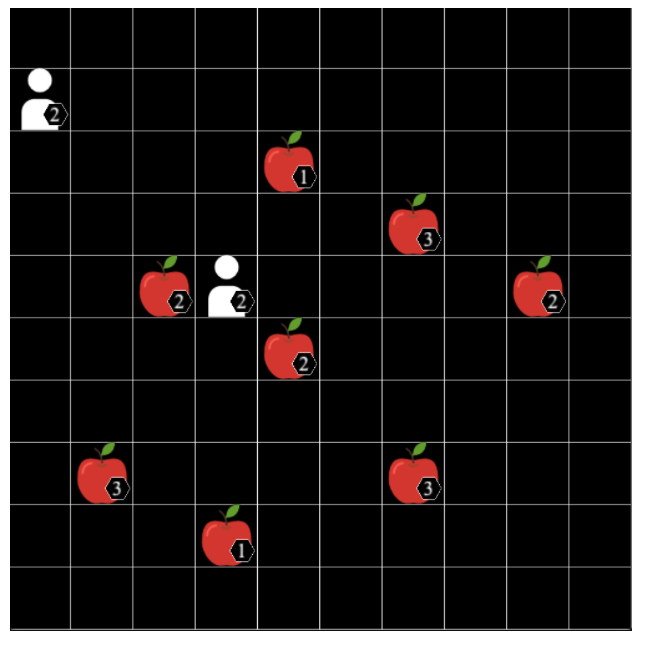}\vspace{-0.5 mm}
\includegraphics[width=.32\linewidth]{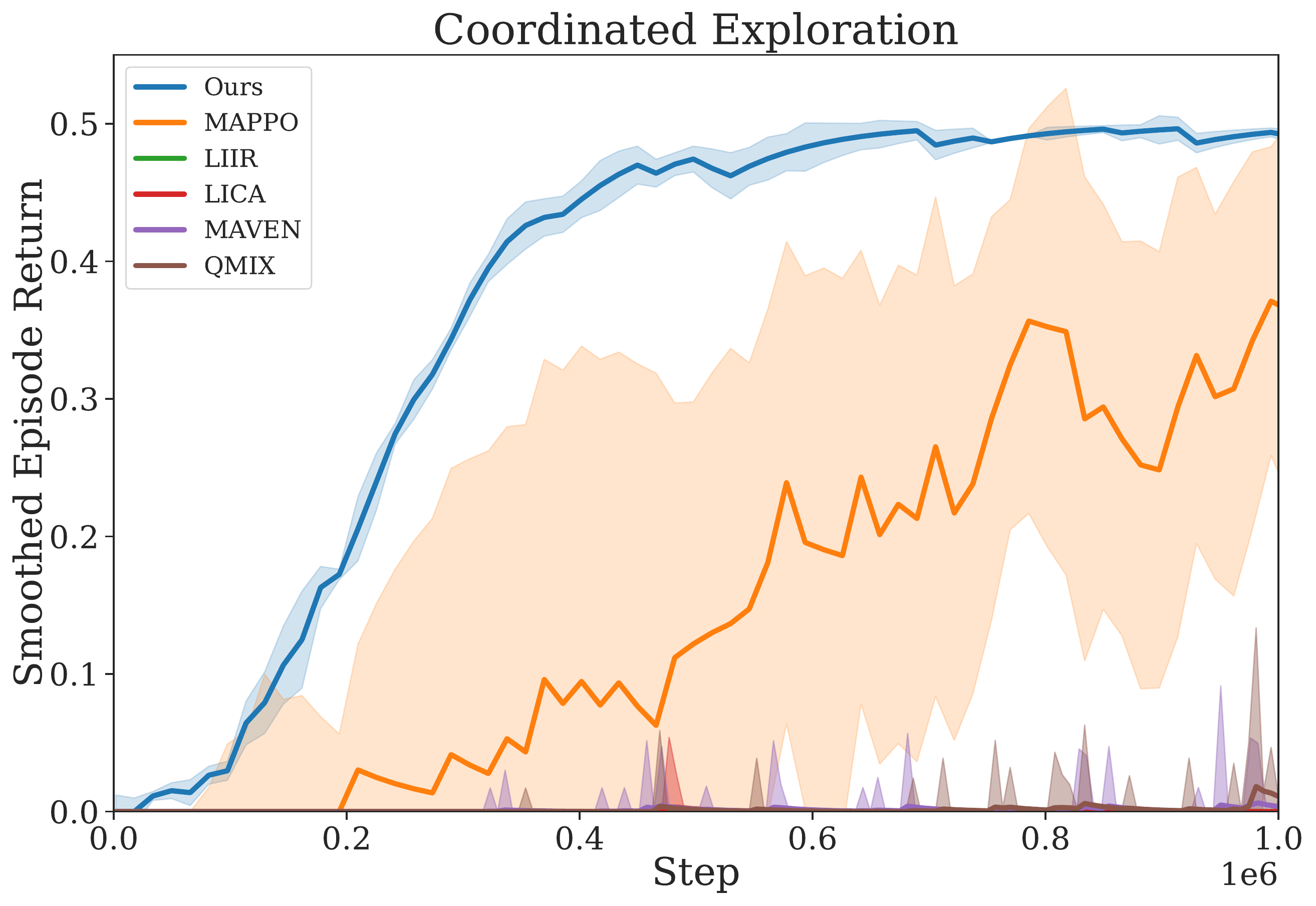}
\includegraphics[width=.32\linewidth]{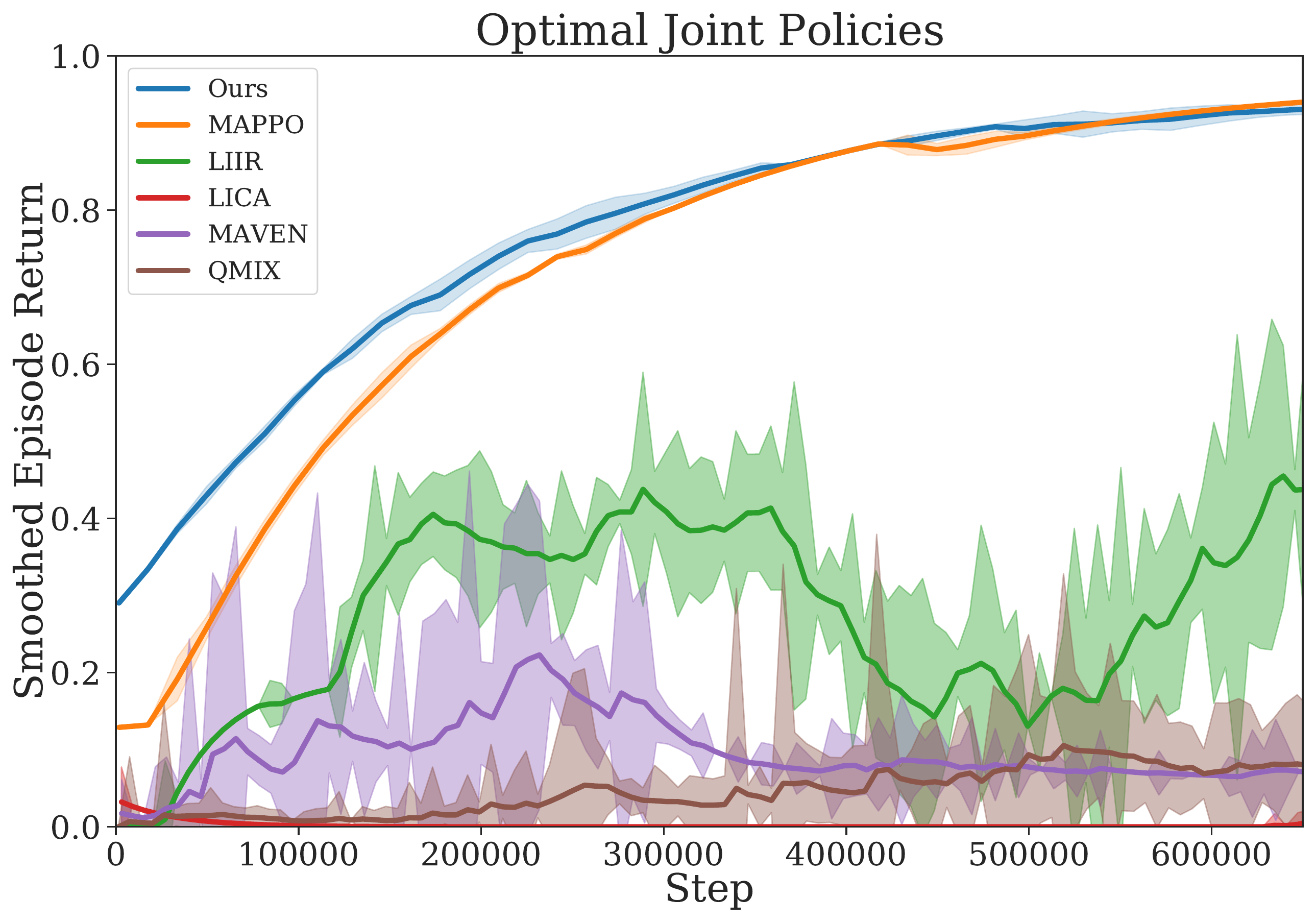}
\includegraphics[width=.32\linewidth]{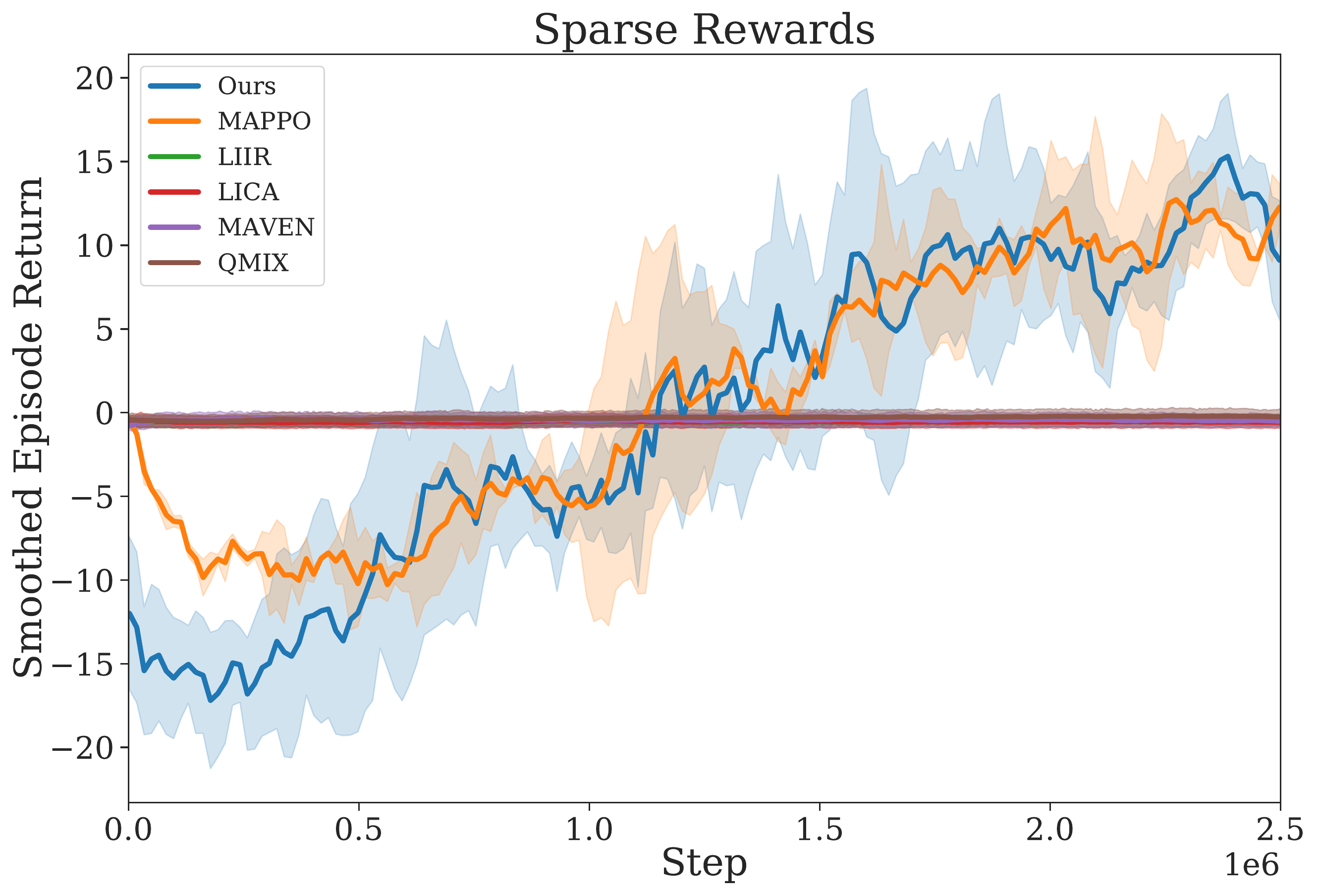}
\captionof{figure}{\textit{Left.} Coordinated Exploration. \textit{Centre}. Optimal joint policies. \textit{Right.} Sparse rewards.}
\vspace{-2mm}
\label{fig:foraging}
\end{figure}

\textbf{Experiment 1: Coordinated exploration.} We tested our first claim that LIGS promotes coordinated exploration among agents. To investigate this, we used a version of the level-based foraging environment \citep{papoudakis2020comparative} as follows: there are $n$ agents each with level $a_i$. Moreover, there are 3 apples with level $K$ such that $\sum_{i=1}^N a_i = K$. The only way to collect the reward is if all agents collectively enact the {\fontfamily{cmss}\selectfont collect} action when they are beside an apple. This is a challenging joint-exploration problem since to obtain the reward, the agents must collectively explore joint actions across the state space (rapidly) to discover that simultaneously executing {\fontfamily{cmss}\selectfont collect} near an apple produces rewards. To increase the difficulty, we added a penalty for the agents failing to coordinate in collecting the apples. For example, if only one agent uses the {\fontfamily{cmss}\selectfont collect} action near an apple, it gets a negative reward. This results in a non-monotonic reward structure. 
The performance curves are given in Fig. \ref{fig:foraging} which shows LIGS demonstrates superior performance over the baselines.

\textbf{Experiment 2: Optimal joint policies. }
We next tested our second claim that LIGS can promote convergence to joint policies that achieve higher system rewards. To do this, we constructed a challenging experiment in which the agents must avoid converging to suboptimal policies that deliver positive but low rewards. In this experiment, the grid is divided horizontally in three sections; top, middle and bottom. All grid locations in the top section give a small reward $r/n$ to the agent visiting them where $n$ is the number of tiles in the each section. The middle section does not give any rewards. The bottom section rewards the agents depending on their relative positions. If one agent is at the top and the other at the bottom, the agent at the bottom receives a reward $-r/n$ each time the other agent receives a reward. If both agents are at the bottom, then one of the tiles in this section will give a reward $R, r/2<R<r$ to both agents. The bottom section gives no reward otherwise. The agents start in the middle section and as soon as they cross to one section they cannot return to the middle. As is shown in Fig. \ref{fig:foraging}, LIGS learns to acquire rewards rapidly in comparison to the baselines with MAPPO requiring around 400k episodes to match the rewards produced by LIGS. 
%
%
%

\textbf{Experiment 3: Sparse rewards. }
We tested our claim that LIGS can promote learning in MAS with sparse rewards. We simulate a sparse reward setting using a competitive game between two teams of agents. One team is controlled by LIGS while the other actions of the agents belonging to the other team are determined by a fixed policy. The goal is to collect the apple faster than the opposing team. Collecting the apple results in a reward of 1, and rewards are 0 otherwise. This is a challenging sparse reward since informative reward signals occur only apple when the apple is collected. As is shown in Fig. 1 both LIGS and MAPPO perform well on the sparse rewards environment, whilst the other baselines are all unable to learn any behaviour on this environment. %
\begin{figure}[h!]
\centering
\includegraphics[width=4.75cm, height=3.5cm]{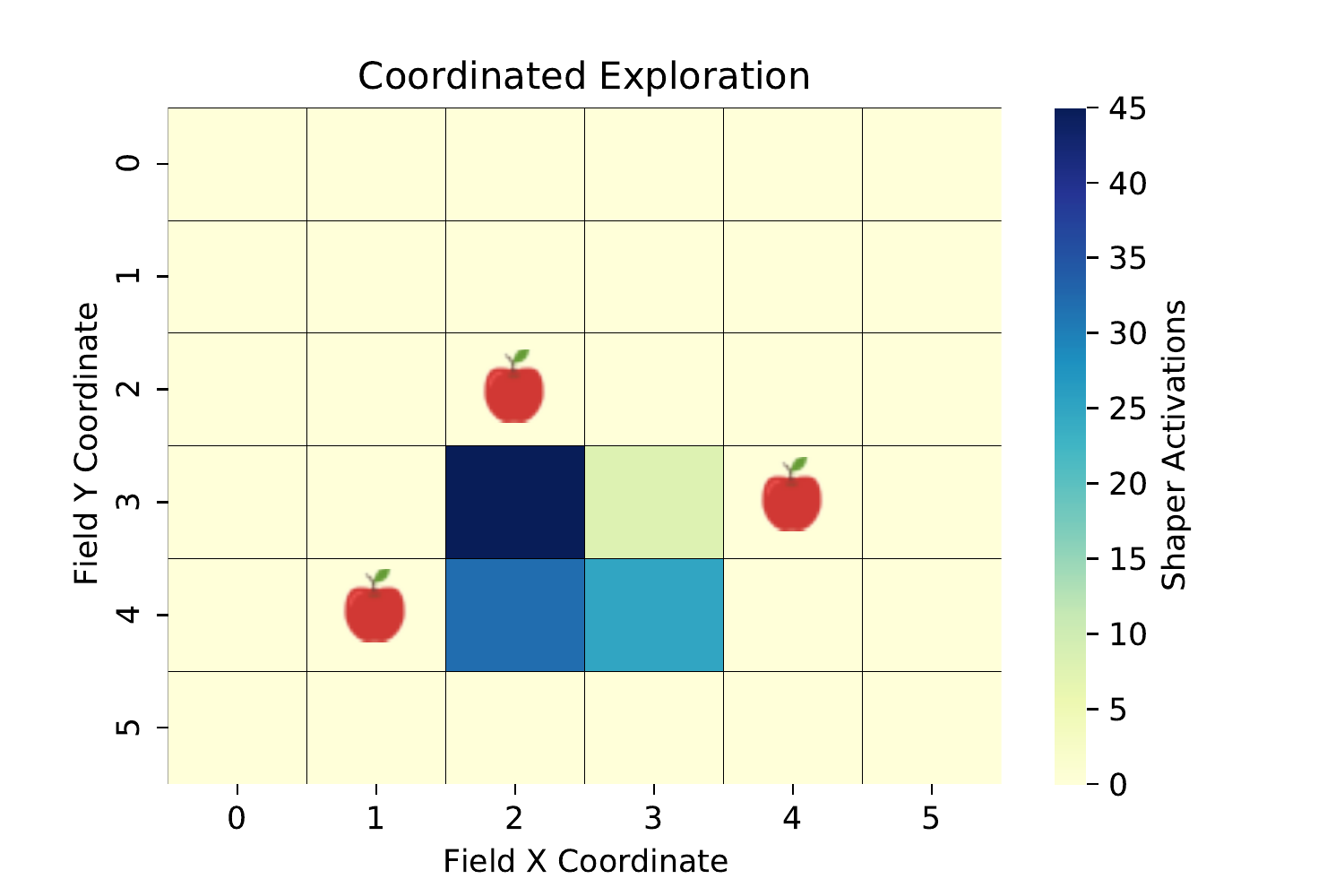}\hspace{1 mm}
\includegraphics[width=4cm, height=3.5cm]{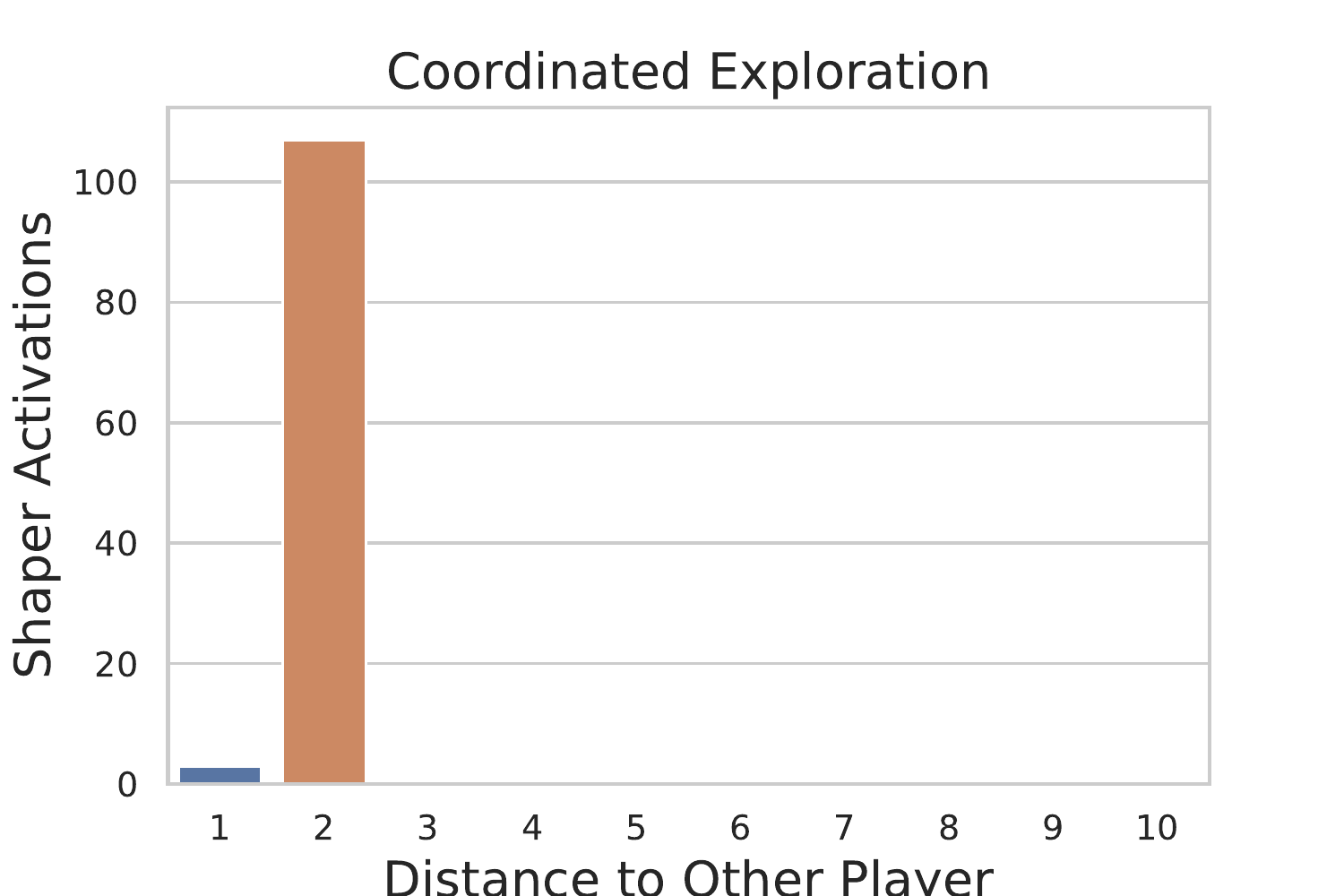}\hspace{1 mm}
\includegraphics[width=4.75cm, height=3.5cm]{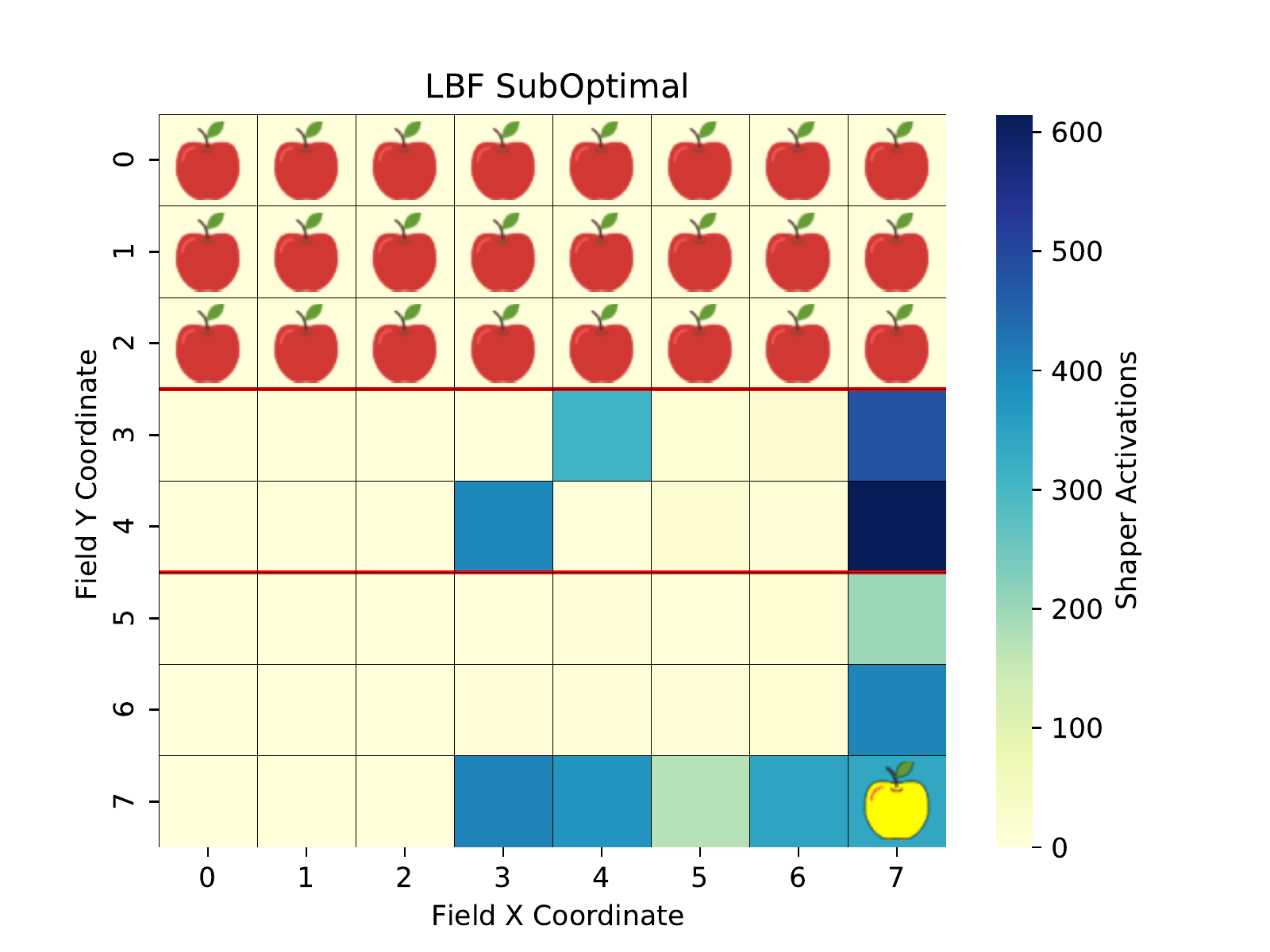}
\captionof{figure}{\textit{Left.} Heatmap of Exp. 1 showing where {\fontfamily{cmss}\selectfont Generator} adds rewards. \textit{Centre.} Plot of distance to other agent when {\fontfamily{cmss}\selectfont Generator} activates rewards in Exp 1. \textit{Right.} Corresponding heatmap for Exp. 2.}
\vspace{-1mm}
\label{fig:heatmap}
\end{figure} 

\textbf{Investigations. }
We investigated the workings of the LIGS framework. We studied the locations where the {\fontfamily{cmss}\selectfont Generator} added intrinsic rewards in Experiments 1 and 2. As  shown in the heatmap visualisation in Fig. \ref{fig:heatmap}, for Experiment 2, we observe that the {\fontfamily{cmss}\selectfont Generator} learns to add intrinsic rewards that guide the agents towards the optimal reward (bottom right) and away from the suboptimal rewards at the top (where some other baselines converge). This supports our claim that LIGS learns to guide the agents towards jointly optimal policies. Also, as Fig. \ref{fig:heatmap} shows, LIGS's switching mechanism means that the {\fontfamily{cmss}\selectfont Generator} only adds intrinsic rewards at the most useful locations for guiding the agents towards their target. For Experiment 1, Fig. \ref{fig:heatmap} shows that the {\fontfamily{cmss}\selectfont Generator} learns to guide the agents towards the apple which delivers the high rewards. Fig. \ref{fig:heatmap} (Centre) demonstrates a striking behaviour of the LIGS framework - it only activates the intrinsic rewards around the apple when \textit{both} agents are at most 2 cells away from the apple. Since the agents receive positive rewards only when they arrive at the apple simultaneously, this ensures the agents are encouraged to coordinate their arrival and receive the maximal rewards and avoids encouraging arrivals that lead to penalties.    
\subsection{Learning Performance in StarCraft Multi-Agent Challenge}
\begin{figure}
\centering
    \includegraphics[scale=0.3]{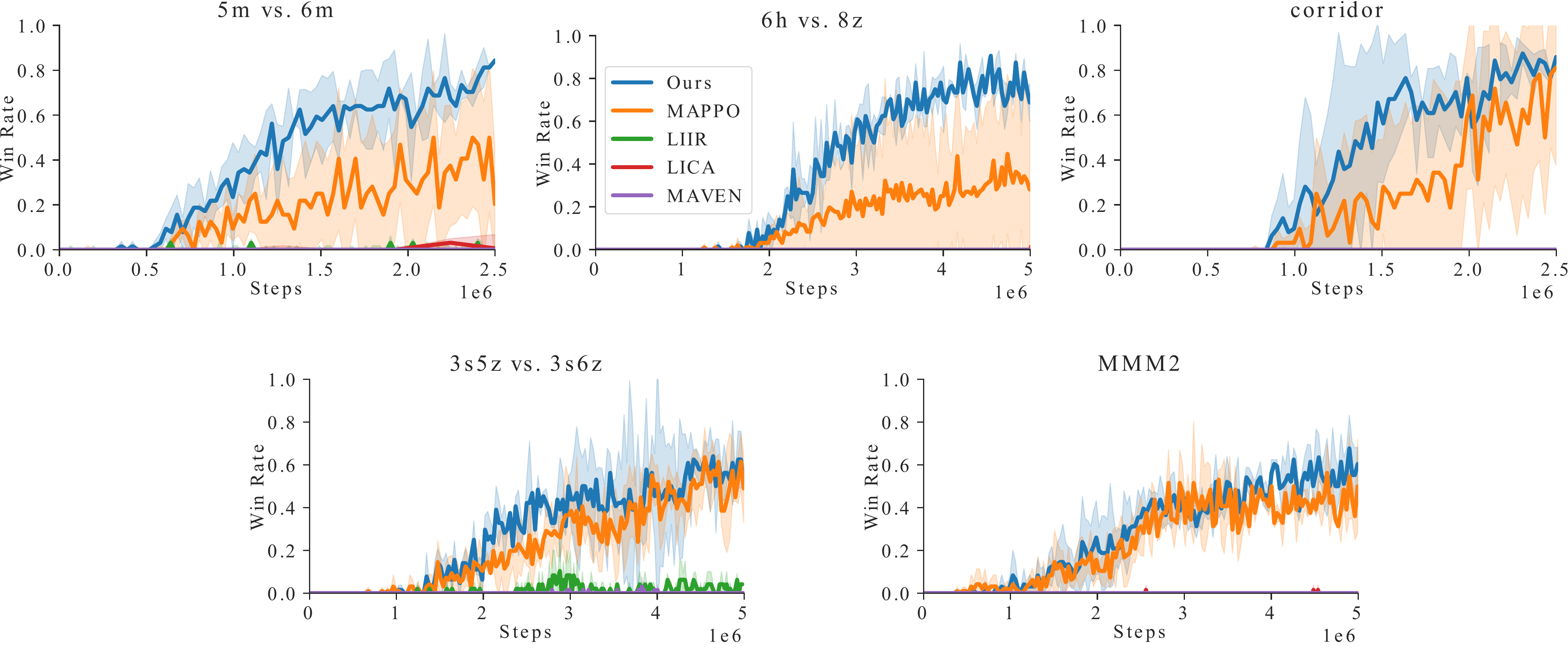}
    \caption{\emph{Median win rate over the course of learning on SMAC.} LIGS outperforms the baselines on all maps. LIIR, LICA, and MAVEN are generally not visible as their win rate is negligible.}
    \label{Figure:SMAC_learning_curves}\vspace{-5 mm}
\end{figure}


%
%

To ascertain if LIGS is effective even in complex environments, we ran it on on the following SMAC maps \emph{5m vs. 6m} (hard), \emph{6h vs. 8z}, \emph{Corridor}, \emph{3s5z vs 3s6z} and \emph{MMM2} (super hard). 
These maps vary in a range of MARL attributes such as number of units to control, environment reward density, unit action sets, and (partial)-observability. In Fig. \ref{Figure:SMAC_learning_curves}, we report our results showing `Win Rate' vs `Steps'. These curves are generated by computing the median win rate (vs the opponent) of the agent at regular intervals during learning. We ran $3$ seeds of each algorithm (further setup details are in the Supp. material Sec \ref{sec:app_imp_details}). LIGS outperforms the baselines in all maps. In \emph{5m vs. 6m} and \emph{6h vs. 8z}, the baselines do not approach the performance of LIGS. In  \emph{Corridor} MAPPO requires over an extra million steps to match LIGS. In \emph{3s5z vs. 3s6z} and \emph{MMM2}, LIGS still outperforms the baselines. In summary, LIGS shows performance gains over all baselines in SMAC maps which encompass diverse MAS attributes.



\vspace{-2mm}
\section{Conclusion}

We introduced LIGS, a novel framework for generating intrinsic rewards which significantly boosts performance of MARL algorithms. Central to LIGS is a powerful adaptive learning mechanism that generates intrinsic rewards according to the task and the MARL learners' joint behaviour. 
%
%
Our experiments show LIGS induces superior performance in MARL algorithms in a range of tasks.
\newpage
\section{Acknowledgements}
We would like to thank Matthew Taylor and Aivar Sootla for their helpful comments. 

%

\bibliographystyle{iclr2022_conference}
\bibliography{main}

\clearpage

\addcontentsline{toc}{section}{Appendix} 
\part{{\Large{Appendix}}} 
\parttoc

\newpage
\section{Algorithm}\label{sec:algorithm}

\begin{algorithm}[h]
    \label{algo:Opt_reward_shap} 
    \DontPrintSemicolon
    \KwInput{ Environment $E$ \;
             \hspace{3em} Initial {\fontfamily{cmss}\selectfont agent} policies $\boldsymbol{\pi}_0=(\pi^1_0,\ldots \pi^N_0)$ with parameters  $\theta_{\pi^1_0},\ldots\theta_{\pi^1_N}$,  Initial {\fontfamily{cmss}\selectfont Generator} switch policy $\mathfrak{g}_{c_{0}}$ with parameters $\theta_{\mathfrak{g}_{c_{0}}}$, Initial {\fontfamily{cmss}\selectfont Generator} action policy $g_0$ with parameters $\theta_{g_0}$,  Randomly initialised fixed neural network $\phi(\cdot, \cdot)$, Neural networks $h$ (fixed) and $\hat{h}$ for Augmented RND with parameter $\theta_{\hat{h}}$, Buffer $B$, Number of rollouts $N_r$, rollout length $T$, Number of mini-batch updates $N_u$, Switch cost $c$, discount factor $\gamma$, learning rate $\alpha$.\;
             }
    \KwOutput{Optimised {\fontfamily{cmss}\selectfont agent} policies $\boldsymbol{\pi^\star}=(\pi^{\star,1},\ldots,\pi^{\star,N})$}
    $\boldsymbol{\pi}=(\pi^1,\ldots,\pi^N), g, \mathfrak{g}_{c} \gets \boldsymbol{\pi}_0, g_0,\mathfrak{g}_{c_{0}}$\;
    \For{$n = 1, N_r$}
    {
        \textbf{// Collect rollouts}\;
        \For{$t = 1, T$}
        {
            Get environment states $s_t$ from $E$ \;
            Sample $\boldsymbol{a}_t=(a^1_t,\ldots,a^N_t)$ from $(\pi^1(s_t),\ldots,\pi^N(s_t))$ \;
            Apply action $\boldsymbol{a}_t$ to environment $E$, get rewards $\boldsymbol{r}_t=(r^1_t,\ldots,r^N_t)$ and next state $s_{t+1}$ \;
            Sample $q_t$ from $\mathfrak{g}_{c}(s_t)$ \textbf{ // Switching control} \;
            \eIf{$q_t = 1$}
            {   
                Sample $\theta^c_t$ from $g(s_t)$ \;
                Sample $\theta^c_{t+1}$ from $g(s_{t+1})$ \;
                $f^i_t = \gamma \theta^c_{t+1} - \theta^c_{t}$ \textbf{// Calculate $F(\theta^c_t, \theta^c_{t+1})$}\;
            }
            {   
                $\theta^c_t, f^i_t = 0, 0$ \textbf{// Dummy values}
            }
            Append $(s_t, \boldsymbol{a}_t, g_t, \theta^c_t, \boldsymbol{r}_t, f^i_t, s_{t+1})$ to $B$
        }
        \For{$u = 1, N_u$}
        {
            Sample data $(s_t, \boldsymbol{a}_t, g_t, \theta^c_t, \boldsymbol{r}_t, f^i_t, s_{t+1})$ from $B$\;
            \eIf{$g_t = 1$}
            {   
                Set reward to $\boldsymbol{r}_t^s = \boldsymbol{r}_t + f^i_t$
            }
            {
                Set reward to $\boldsymbol{r}_t^s = \boldsymbol{r}_t$
            }
            \textbf{// Update Augmented RND} \;
            $\text{Loss}_{\text{RND}} = ||h(s_t,\boldsymbol{a}_t) - \hat{h}(s_t,\boldsymbol{a}_t)||^2$ \;
            $\theta_{\hat{h}} \gets \theta_{\hat{h}} - \alpha \nabla \text{Loss}_{\text{RND}} $ \;
            \textbf{// Update {\fontfamily{cmss}\selectfont Generator} }\;
            $l_t = ||h(s_t,\boldsymbol{a}_t) - \hat{h}(s_t)||^2$ \textbf{// Compute $L(s_t,\boldsymbol{a}_t)$} \;
            $c_t = c g_t$ \;
            Compute $\text{Loss}_{g}$ using $(s_t, a_t, g_t, c_t, \boldsymbol{r}_t, f^i_t, l_t, s_{t+1})$ using PPO loss \textbf{// Section \ref{sec:learning_proc}} \;
            Compute $\text{Loss}_{ \mathfrak{g}_{c}}$ using $(s_t, a_t, g_t, c_t, \boldsymbol{r}_t, f^i_t, l_t, s_{t+1})$ using PPO loss \textbf{// Section \ref{sec:learning_proc}}\;
            $\theta_{g} \gets \theta_{g} - \alpha \nabla \text{Loss}_{g}$ \;
            $\theta_{\mathfrak{g}_{c}} \gets \theta_{\mathfrak{g}_{c}} - \alpha \nabla \text{Loss}_{\mathfrak{g}_{c}}$ \;
            \textbf{// Update agent $j$, for each $ j \in 1,\ldots, N$}\;
            Compute $\text{Loss}_{ \pi^j}$ using $(s_t, \boldsymbol{a}_t, r^{j,s}_t:=r^j_t+f^i_t, s_{t+1})$ using PPO loss \textbf{// Section \ref{sec:learning_proc}} \;
            $\theta_{\pi^j} \gets \theta_{\pi^j} - \alpha \nabla \text{Loss}_{\pi^j}$ \;
        }
    }
	\caption{\textbf{L}earnable \textbf{I}ntrinsic-Reward \textbf{G}eneration \textbf{S}election algorithm (LIGS)}
\end{algorithm}
\newpage

\section{Ablation study: Plug \& Play}\label{sec:plug_n_play}
In order to validate our claim that LIGS freely adopts RL learners, we tested the ability of LIGS to boost performance in a complex coordination task using independent Proximal policy optimization algorithm (IPPO) \citep{schulman2017Proximal} as the base learner. In this experiment, two agents are spawned at opposite sides of the grid. The red agent is spawned in the left hand side and the blue agent is spawned in the right hand side of the grid in Fig. \ref{fig:my_label} (right). The goal of the agents is to arrive at their corresponding goal states (indicated by the coloured square, where the colour corresponds to the agent whose goal state it is) at the other side of the grid. Upon arriving at their goal state the agents receive their reward. However, the task is made difficult by the fact that only one agent can pass through the corridor at a time. Therefore, in this setup, the only way for the agents to complete the task is for the agents to successfully coordinate, i.e. one agent is required to allow the other agent to pass through before attempting to traverse the corridor. 

It is known that independent learners in general, struggle to solve such tasks since their ability to coordinate systems of RL learners is lacking \citep{yang2020multi}.  This is demonstrated in Fig. \ref{fig:my_label} (left) which displays the performance curve of for IPPO which fails to score above $0$. As claimed, when incorporated into the LIGS framework, the agents succeed in coordinating to solve the task. This is indicated by the performance of IPPO + LIGS (blue).  

\begin{figure}[h!]
    \centering
\includegraphics[width=0.4\textwidth]{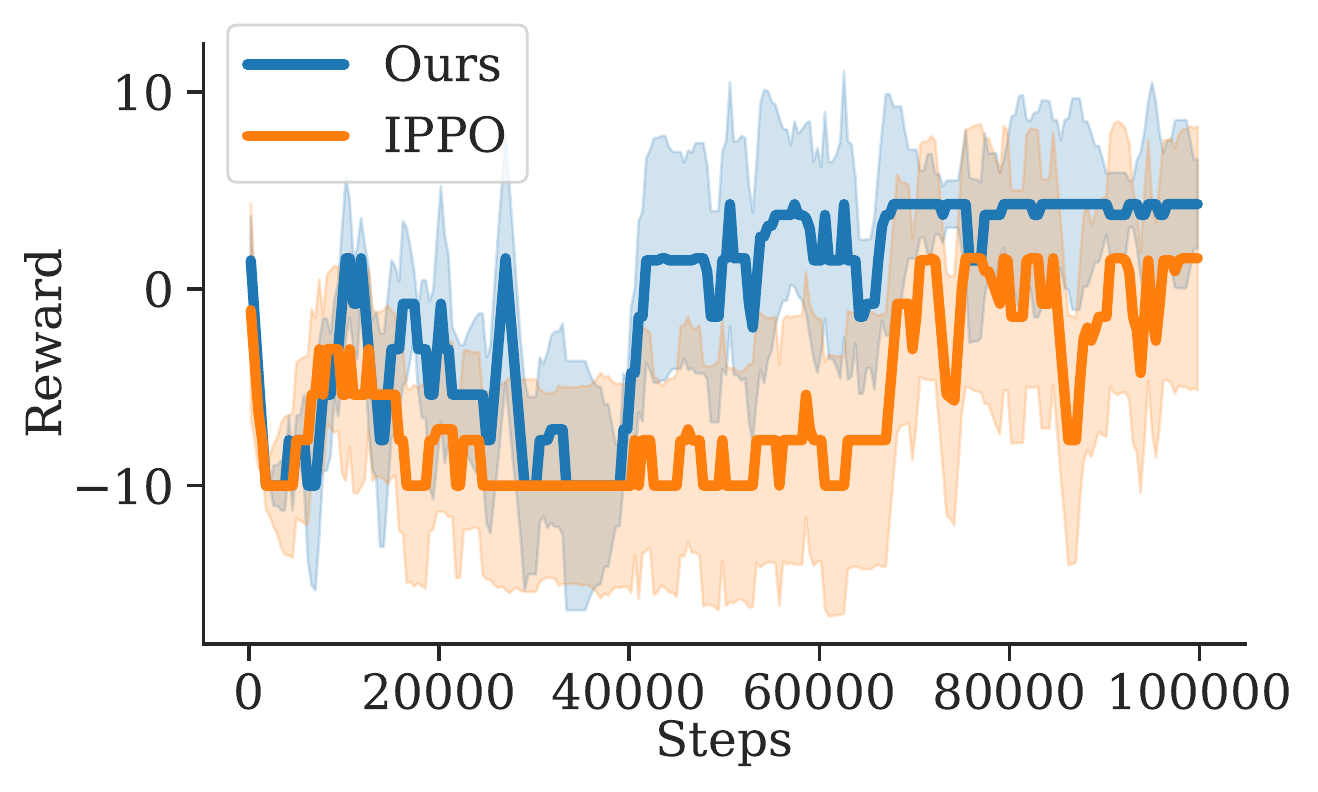}
\vspace{3mm}\hspace{10 mm}
    \includegraphics[width=0.4\textwidth]{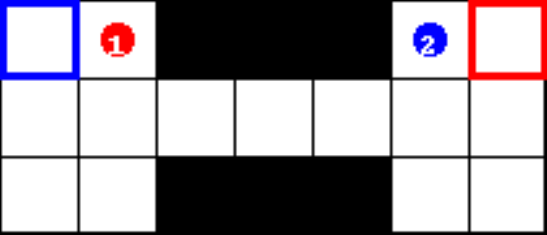}
    \caption{\textit{Left.} Performance curves for IPPO and IPPO with LIGS. \textit{Right.} Coordination environment.  }
    \label{fig:my_label}
\end{figure}

\newpage
\section{Ablation Study: The Utility of Switching Controls}
A core component of LIGS is the switching control mechanism. This component enables the {\fontfamily{cmss}\selectfont Generator} to selectively add intrinsic rewards only at the set of states most relevant for improving learning outcomes while avoiding adding intrinsic rewards where they are not necessary. To evaluate the impact of this component of LIGS, we compared the performance of LIGS with a version in which the switching control was replaced with an equal-chances Bernoulli Random Variable (i.e., at any given state, the {\fontfamily{cmss}\selectfont Generator} adds or does not add intrinsic rewards with equal probability), and, a version where it always adds intrinsic rewards. Figure \ref{fig:ablation_switching_controls} shows the performance of these three versions of LIGS. We added vanilla MAPPO as a baseline reference. We examined the performance of the variants of LIGS on the coordination task described in Section \ref{sec:plug_n_play}. As can be seen in the plot, incorporating learned switching controls in LIGS (labelled "LIGS") leads to superior performance compared to simply adding intrinsic rewards at random (line labelled "LIGS with Random Switching") and adding intrinsic rewards everywhere (labelled "LIGS with Always Adding intrinsic Rewards"). In fact, adding intrinsic rewards at random is detrimental to performance as demonstrated by the fact that the performance of LIGS with Random Switching is worse than that of vanilla MAPPO.

\begin{figure}[h!]
    \centering
    \includegraphics[width=\textwidth]{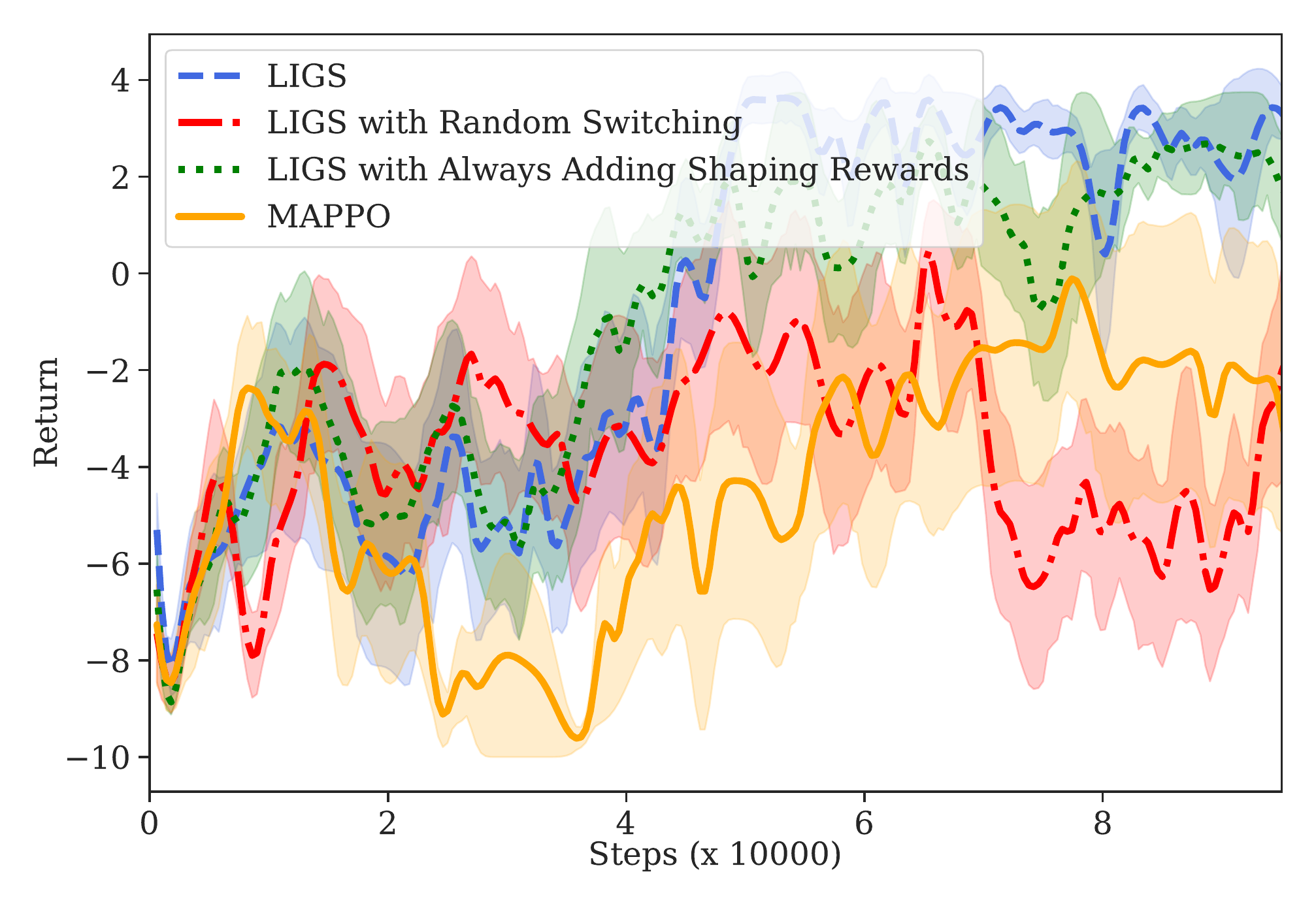}
    \caption{Ablation of the switching control mechanism. Learned switching controls ("LIGS") outperform versions where intrinsic rewards are added at random ("LIGS with Random Switching") and where intrinsic rewards are always added ("LIGS with Always Added intrinsic Rewards").}
    \label{fig:ablation_switching_controls}
\end{figure}



\newpage
\section{Flexibility of LIGS to Accommodate different Exploration Bonus Terms $L$}
To demonstrate the robustness of our method to different choices of exploration bonus terms in {\fontfamily{cmss}\selectfont Generator}'s objective, we conducted an Ablation study on the $L$-term (c.f. Equation \ref{generator_objective}) where we replaced the RND $L$ term with a basic count-based exploration bonus. To exemplify the high degree of flexibility, we replaced the RND with a simple exploration bonus term  $L(s)=\frac{1}{\text{Count(s)}+1}$ for any given state $s\in\mathcal{S}$ where Count$(s)$ refers to a simple count of the number of times the state $s$ has been visited.  We conducted the Ablation study on all three Foraging environments presented in Sec. \ref{exp:foraging}. We note that despite the simplicity of the count-based measure, generally the performance of both versions of LIGS is comparable and in fact the count-based variant is superior to the RND version for the joint exploration environment.

\begin{figure}[h!]
    \centering
    \includegraphics[width=0.45\textwidth]{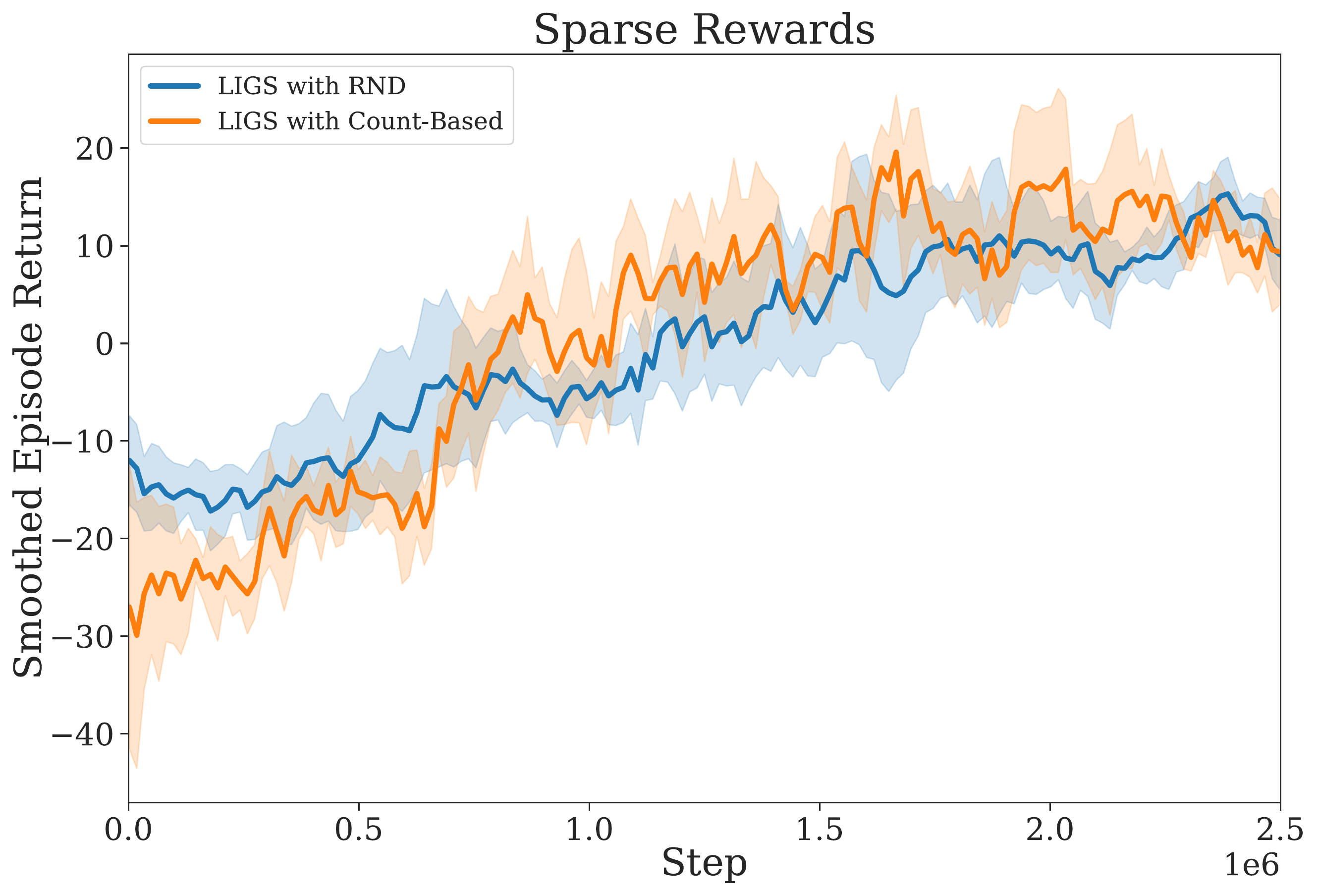}
    \includegraphics[width=0.45\textwidth]{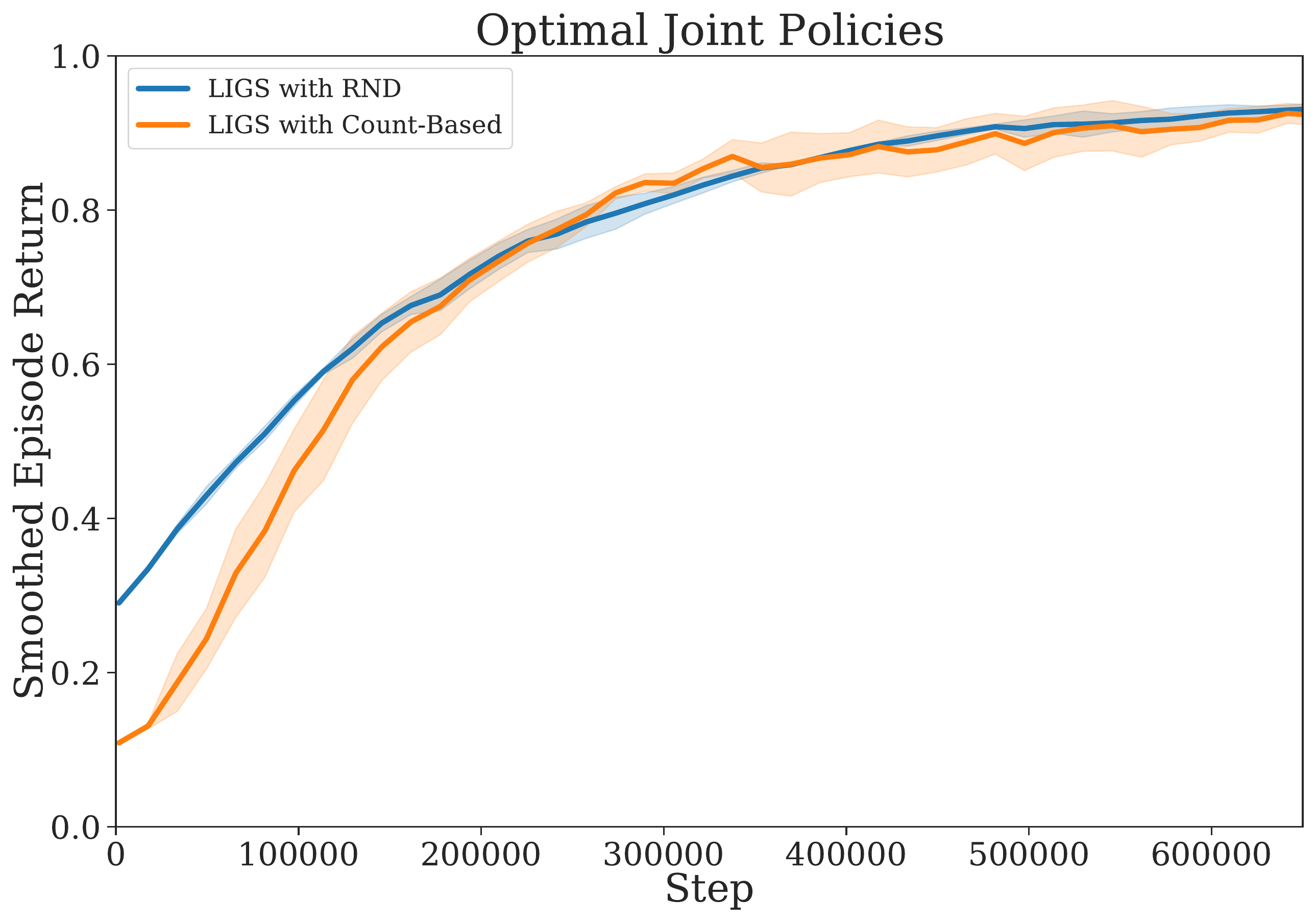}
        \includegraphics[width=0.45\textwidth]{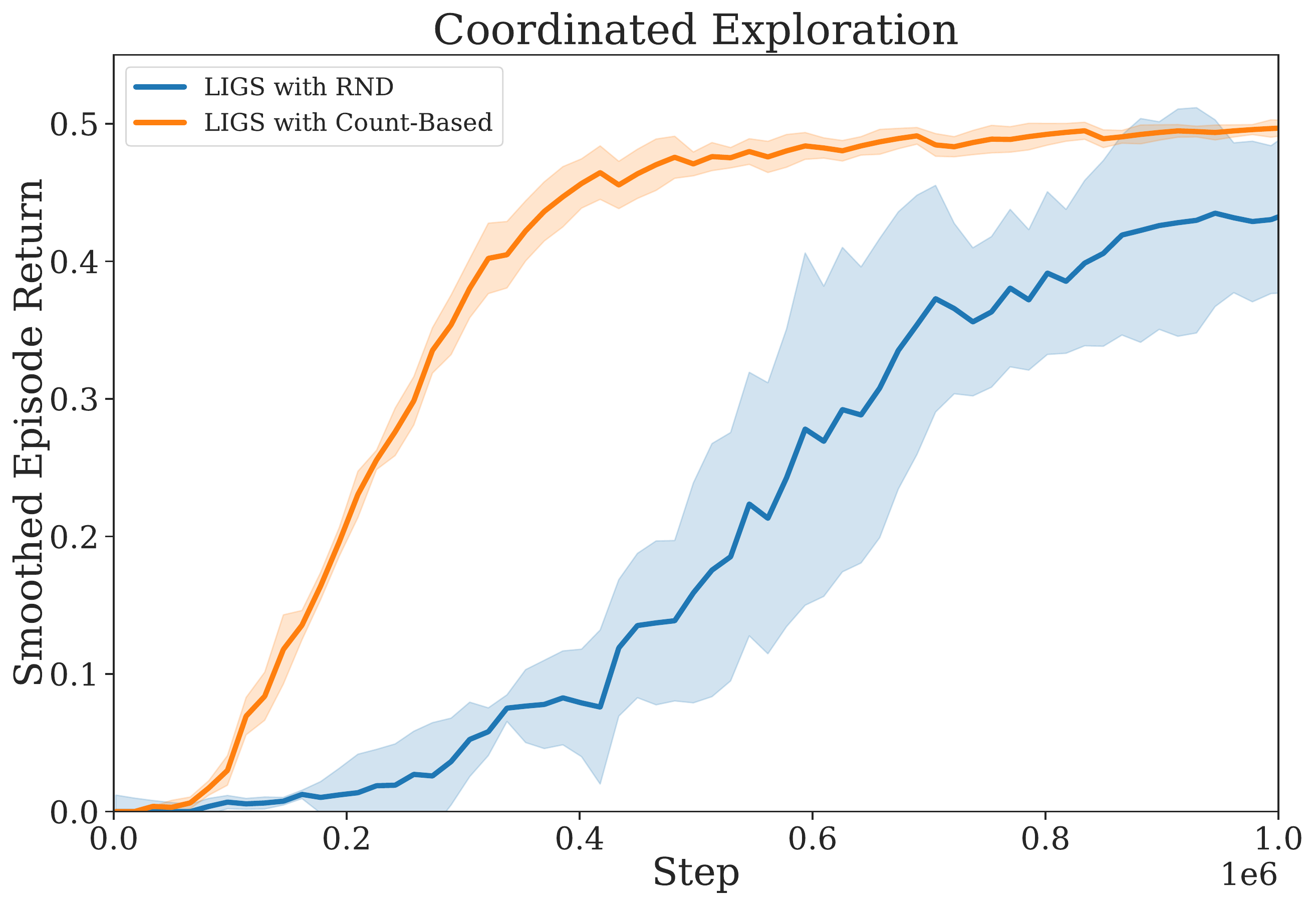}
    \caption{Performance of LIGS compared with  the exploration bonus replaced by count-based method on the three tasks in the Foraging environment.}
    \label{fig:count_based_comparison}
\end{figure}



\newpage
\section{Further Experiment Demonstrating LIGS improved use of Exploration Bonuses.}
As we have shown above, LIGS can accommodate a variety of exploration bonuses and perform well. Here, we did a experiment to further justify using LIGS against simpler exploration bonus methods. We compared LIGS against and MAPPO with an RND  intrinsic reward in the agents' objectives (MAPPO+RND) and vanilla MAPPO. Fig. \ref{fig:us_vs} shows performance of these two methods on coordination environment shown in Fig. \ref{fig:my_label}. We note that LIGS markedly outperforms both MAPPO+RND and vanilla MAPPO. Due to the added benefit of switching controls and intrinsic reward selection performed by the {\fontfamily{cmss}\selectfont Generator}, we observe that LIGS is able to significantly augment the benefits of applying RND directly to the agents' objectives.

\begin{figure}[h!]
    \centering
\includegraphics[width=0.8\textwidth]{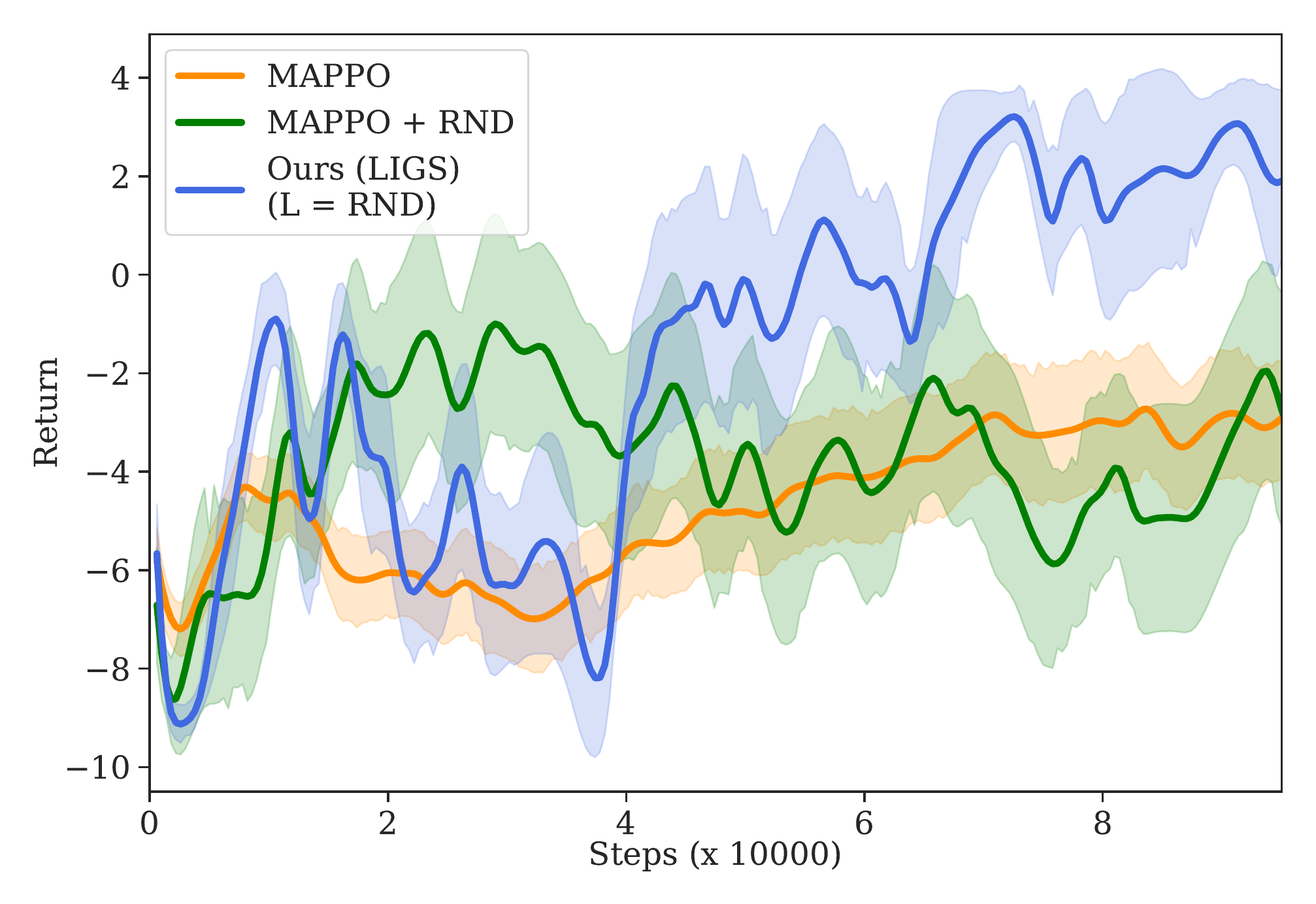}
    \caption{Performance curves for LIGS, MAPPO with RND intrinsic rewards and vanilla MAPPO. The additional machinery of switching-controls and intrinsic reward selection allows LIGS to make better use of exploration bonuses. In this case, LIGS demonstrates significant improvement over MAPPO with RND intrinsic rewards.}
    \label{fig:us_vs}
\end{figure}

\newpage
\section{Further Implementation Details}\label{sec:app_imp_details}


Details of the {\fontfamily{cmss}\selectfont Generator} and $F$ (intrinsic-reward)\\
\begin{tabular}{c|l}
\textbf{Object} & \textbf{Description}\\
\hline
$\Theta$ & Discrete action set which is size of output of $f$,\\
    & i.e., $\Theta$ is set of integers $\{1,...,m\}$ \\
$g$ & Fixed feed forward NN that maps $\mathbb{R}^d \mapsto \mathbb{R}^m$ \\
    & [512, \texttt{ReLU}, 512, \texttt{ReLU}, 512, $m$] \\
$F$ & $\gamma \theta^c_{t+1}$ - $\theta^c_{t},$\;\; $\gamma=0.95$\\
\end{tabular}

$d$=Dimensionality of states; $m\in \mathbb{N}$ - tunable free parameter.

\underline{In all experiments} we used the above form of $F$ as follows: a state $s_t$ is input to the $g$ network and the network outputs logits $p_t$. 
we softmax and sample from $p_t$ to obtain the action $\theta^c_t$. This action is one-hot encoded. 
In this way the policy of the {\fontfamily{cmss}\selectfont Generator} chooses the intrinsic-reward.

\subsection{Hyperparameter Settings}
In the table below we report all hyperparameters used in our experiments. Hyperparameter values in square brackets indicate ranges of values that were used for performance tuning.

\begin{center}
    \begin{tabular}{c|c} 
        \toprule
        Clip Gradient Norm & 1\\
        $\gamma_{E}$ & 0.99\\
        $\lambda$ & 0.95\\
        Learning rate & $1$x$10^{-4}$ \\
        Number of minibatches & 4\\
        Number of optimisation epochs & 4\\
        Number of parallel actors & 16\\
        Optimisation algorithm & ADAM\\
        Rollout length & 128\\
        Sticky action probability & 0.25\\
        Use Generalized Advantage Estimation & True\\
        \midrule
        Coefficient of extrinsic reward & [1, 5]\\
        Coefficient of intrinsic reward & [1, 2, 5, 10, 20, 50]\\
        {\fontfamily{cmss}\selectfont Generator} discount factor & 0.99\\
        Probability of terminating option & [0.5, 0.75, 0.8, 0.9, 0.95]\\
        $L$ function output size & [2, 4, 8, 16, 32, 64, 128, 256]\\
        \bottomrule
    \end{tabular}
\end{center}
\clearpage
\section{Notation \& Assumptions}\label{sec:notation_appendix}

We assume that $\mathcal{S}$ is defined on a probability space $(\Omega,\mathcal{F},\mathbb{P})$ and any $s\in\mathcal{S}$ is measurable with respect
to the Borel $\sigma$-algebra associated with $\mathbb{R}^p$. We denote the $\sigma$-algebra of events generated by $\{s_t\}_{t\geq 0}$
by $\mathcal{F}_t\subset \mathcal{F}$. In what follows, we denote by $\left( \mathcal{V},\|\|\right)$ any finite normed vector space and by $\mathcal{H}$ the set of all measurable functions.  Where it will not cause confusion (and with a minor abuse of notation) for a given function $h$ we use the shorthand $h^{(\pi^{i},\pi^{-i})}(s)= h(s,\pi^i,\pi^{-i})\equiv\mathbb{E}_{\pi^i,\pi^{-i}}[h(s,a^i,a^{-i})]$.

The results of the paper are built under the following assumptions which are standard within RL and stochastic approximation methods:

\textbf{Assumption 1}
The stochastic process governing the system dynamics is ergodic, that is  the process is stationary and every invariant random variable of $\{s_t\}_{t\geq 0}$ is equal to a constant with probability $1$.

\textbf{Assumption 2}
The constituent functions of the agents' objectives $R$, $F$ and $L$ are in $L_2$.

\textbf{Assumption 3}
For any positive scalar $c$, there exists a scalar $\mu_c$ such that for all $s\in\mathcal{S}$ and for any $t\in\mathbb{N}$ we have: $
    \mathbb{E}\left[1+\|s_t\|^c|s_0=s\right]\leq \mu_c(1+\|s\|^c)$.

\textbf{Assumption 4}
There exists scalars $C_1$ and $c_1$ such that for any function $J$ satisfying $|J(s)|\leq C_2(1+\|s\|^{c_2})$ for some scalars $c_2$ and $C_2$ we have that: $
    \sum_{t=0}^\infty\left|\mathbb{E}\left[J(s_t)|s_0=s\right]-\mathbb{E}[J(s_0)]\right|\leq C_1C_2(1+\|s_t\|^{c_1c_2})$.

\textbf{Assumption 5}
There exists scalars $c$ and $C$ such that for any $s\in\mathcal{S}$ we have that: $
    |J(s,\cdot)|\leq C(1+\|s\|^c)$ for $J\in \{R,F,L\}$.

We also make the following finiteness assumption on set of switching control policies for the {\fontfamily{cmss}\selectfont Generator}:

\textbf{Assumption 6}
For any policy $\mathfrak{g}_c$, the total number of interventions is $K<\infty$.

We lastly make the following assumption on $L$ which can be made true by construction:

\textbf{Assumption 7}
Let $n(s)$ be the state visitation count for a given state $s\in\mathcal{S}$. For any $\boldsymbol{a}\in\boldsymbol{\mathcal{A}}$, the function $L(s,\boldsymbol{a})= 0$ for any $n(s)\geq M$ where $0<M\leq \infty$.

\clearpage
\section{Proof of Technical Results}\label{sec:proofs_appendix}

We begin the analysis with some preliminary lemmata and definitions which are useful for proving the main results.

Given a $V^{\boldsymbol{\pi},g}:\mathcal{S}\times\mathbb{N}\to\mathbb{R},\;\forall\boldsymbol{\pi}\in\boldsymbol{\Pi}$ and $g$, $\forall s_{\tau_k}\in\mathcal{S}$, we define the {\fontfamily{cmss}\selectfont Generator} intervention operator $\mathcal{M}^{\boldsymbol{\pi},g}V^{\boldsymbol{\pi},g}$ by 
\begin{align}
\mathcal{M}^{\boldsymbol{\pi},g}V^{\boldsymbol{\pi},g}(s_{\tau_k},I_{\tau_k}):=R(s_{\tau_k},\boldsymbol{a}_{\tau_k})+F^{(\theta_{\tau_k},\theta_{\tau_{k-1}})}-\delta^{\tau_k}_{\tau_{k}}+\gamma\sum_{s'\in\mathcal{S}}P(s';\boldsymbol{a}_{\tau_k},s)V^{\boldsymbol{\pi},g}(s',I(\tau_{k+1})), \label{intervention_op}
\end{align} 
where $\boldsymbol{a}_{\tau_k}\sim \boldsymbol{\pi}(\cdot|s_{\tau_k})$, $\theta_{\tau_k}\sim  g(\cdot|s_{\tau_k})$ and $\tau_k$ is a  {\fontfamily{cmss}\selectfont Generator} switching time.
We define the Bellman operator $T$ of $\mathcal{G}$ by 
\begin{align}
T V^{\boldsymbol{\pi},g}(s_t,I_t):=\max\Big\{\mathcal{M}^{\boldsymbol{\pi},g}V^{\boldsymbol{\pi},g}(s_t,I_t),  R(s_t,\boldsymbol{a}_t)+\gamma\underset{\boldsymbol{a}\in\boldsymbol{\mathcal{A}}}{\max}\;\sum_{s'\in\mathcal{S}}P(s';\boldsymbol{a},s_t)V^{\boldsymbol{\pi},g}(s',I_t)\Big\} \label{bellman_op}.
\end{align}

\begin{definition}{A.1}
An operator $T: \mathcal{V}\to \mathcal{V}$ is said to be a \textbf{contraction} w.r.t a norm $\|\cdot\|$ if there exists a constant $c\in[0,1[$ such that for any $V_1,V_2\in  \mathcal{V}$ we have that:
\begin{align}
    \|TV_1-TV_2\|\leq c\|V_1-V_2\|.
\end{align}
\end{definition}

\begin{definition}{A.2}
An operator $T: \mathcal{V}\to  \mathcal{V}$ is \textbf{non-expansive} if $\forall V_1,V_2\in  \mathcal{V}$ we have:
\begin{align}
    \|TV_1-TV_2\|\leq \|V_1-V_2\|.
\end{align}
\end{definition}

\begin{lemma} \label{max_lemma}
For any
$f: \mathcal{V}\to\mathbb{R},g: \mathcal{V}\to\mathbb{R}$, we have that:
\begin{align}
\left\|\underset{a\in \mathcal{V}}{\max}\:f(a)-\underset{a\in \mathcal{V}}{\max}\: g(a)\right\| \leq \underset{a\in \mathcal{V}}{\max}\: \left\|f(a)-g(a)\right\|.    \label{lemma_1_basic_max_ineq}
\end{align}
\end{lemma}
\begin{proof}
We restate the proof given in \cite{mguni2019cutting}:
\begin{align}
f(a)&\leq \left\|f(a)-g(a)\right\|+g(a)\label{max_inequality_proof_start}
\\\implies
\underset{a\in \mathcal{V}}{\max}f(a)&\leq \underset{a\in \mathcal{V}}{\max}\{\left\|f(a)-g(a)\right\|+g(a)\}
\leq \underset{a\in \mathcal{V}}{\max}\left\|f(a)-g(a)\right\|+\underset{a\in \mathcal{V}}{\max}\;g(a). \label{max_inequality}
\end{align}
Deducting $\underset{a\in \mathcal{V}}{\max}\;g(a)$ from both sides of (\ref{max_inequality}) yields:
\begin{align}
    \underset{a\in \mathcal{V}}{\max}f(a)-\underset{a\in \mathcal{V}}{\max}g(a)\leq \underset{a\in \mathcal{V}}{\max}\left\|f(a)-g(a)\right\|.\label{max_inequality_result_last}
\end{align}
After reversing the roles of $f$ and $g$ and redoing steps (\ref{max_inequality_proof_start}) - (\ref{max_inequality}), we deduce the desired result since the RHS of (\ref{max_inequality_result_last}) is unchanged.
\end{proof}

\begin{lemma}{A.4}\label{non_expansive_P}
The probability transition kernel $P$ is non-expansive, that is:
\begin{align}
    \|PV_1-PV_2\|\leq \|V_1-V_2\|.
\end{align}
\end{lemma} 
\begin{proof}
The result is well-known e.g. \citep{tsitsiklis1999optimal}. We give a proof using the Tonelli-Fubini theorem and the iterated law of expectations, we have that:
\begin{align*}
&\|PJ\|^2=\mathbb{E}\left[(PJ)^2[s_0]\right]=\mathbb{E}\left(\left[\mathbb{E}\left[J[s_1]|s_0\right]\right)^2\right]
\leq \mathbb{E}\left[\mathbb{E}\left[J^2[s_1]|s_0\right]\right] 
= \mathbb{E}\left[J^2[s_1]\right]=\|J\|^2,
\end{align*}
where we have used Jensen's inequality to generate the inequality. This completes the proof.
\end{proof}

\section*{Proof of Prop. \ref{preservation_lemma}}
\begin{proof}
To prove (i) of the proposition it suffices to prove that the term $\sum_{t=0}^T\gamma^{t}F(\theta^c_t,\theta^c_{t-1})I(t)$ converges to $0$ in the limit as $T\to \infty$. As in classic potential-based reward shaping \citep{ng1999policy}, central to this observation is the telescoping sum that emerges by construction of $F$.

First recall $\tilde{v}^{\boldsymbol{\pi},\mathfrak{g}}(s,I_0)$, for any $(s,I_0)\in\mathcal{S}\times\{0,1\}$ is given by:
\begin{align}
&\tilde{v}^{\boldsymbol{\pi},\mathfrak{g}}(s,I_0)=\mathbb{E}_{\boldsymbol{\pi},g}\left[\sum_{t=0}^\infty \gamma^t\left\{R(s_t,\boldsymbol{a}_t)+F(\theta^c_t,\theta^c_{t-1})I_t\right\}\right]
\\&=\mathbb{E}_{\boldsymbol{\pi},g}\left[\sum_{t=0}^\infty \gamma^tR(s_t,\boldsymbol{a}_t)+\sum_{t=0}^\infty \gamma^tF(\theta^c_t,\theta^c_{t-1})I_t\right]
\\&=\mathbb{E}_{\boldsymbol{\pi},g}\left[\sum_{t=0}^\infty \gamma^tR(s_t,\boldsymbol{a}_t)\right]+\mathbb{E}_{\boldsymbol{\pi},g}\left[\sum_{t=0}^\infty \gamma^tF(\theta^c_t,\theta^c_{t-1}))I_t\right].
\end{align}

Hence it suffices to prove that $\mathbb{E}_{\boldsymbol{\pi},g}\left[\sum_{t=0}^\infty \gamma^tF(\theta^c_t,\theta^c_{t-1}))I_t\right]=0$.

Recall there a number of time steps that elapse between $\tau_k$ and $\tau_{k+1}$, now
\begin{align*}
&\sum_{t=0}^\infty\gamma^{t}F(\theta^c_t,\theta^c_{t-1}))I(t)
\\&=\sum_{t=\tau_1+1}^{\tau_2}\gamma^{t}\theta^c_t-\gamma^{t-1}\theta^c_{t-1}+\gamma^{\tau_1}\theta^c_{\tau_1} +\sum_{t=\tau_3+1}^{\tau_4}\gamma^{t}\theta^c_t-\gamma^{t-1}\theta^c_{t-1}+\gamma^{\tau_3}\theta^c_{\tau_3}
\\&\quad+\ldots+ \sum_{t=\tau_{(2k-1)}+1}^{\tau_{2k
}}\gamma^{t}\theta^c_t-\gamma^{t-1}\theta^c_{t-1}+\gamma^{\tau_1}\theta^c_{\tau_{2k+1}}+\ldots+
\\&=\sum_{t=\tau_1}^{\tau_2-1}\gamma^{t+1}\theta^c_{t+1}-\gamma^{t}\theta^c_{t}+\gamma^{\tau_1}\theta^c_{\tau_1}+\sum_{t=\tau_3}^{\tau_4-1}\gamma^{t+1}\theta^c_{t+1}-\gamma^{t}\theta^c_{t}+\gamma^{\tau_3}\theta^c_{\tau_3}
\\&\quad+\ldots+ \sum_{t=\tau_{(2k-1)}}^{\tau_{2K-1}}\gamma^{t}\theta^c_t-\gamma^{t-1}\theta^c_{t-1}+\gamma^{\tau_{2k-1}}\theta^c_{\tau_{2k-1}}+\ldots+
\\&=\sum_{k=1}^\infty\sum_{t=\tau_{2k-1}}^{\tau_{2K-1}}\gamma^{t+1}\theta^c_{t+1}-\gamma^{t}\theta^c_{t}-\sum_{k=1}^\infty\gamma^{\tau_{2k-1}}\theta^c_{\tau_{2k-1}}
\\&=\sum_{k=1}^\infty\gamma^{\tau_{2k}}\theta^c_{\tau_{2k}}-\sum_{k=1}^\infty\gamma^{\tau_{2k-1}}\theta^c_{\tau_{2k-1}}
\\&=\sum_{k=1}^\infty\gamma^{\tau_{2k}}0-\sum_{k=1}^\infty\gamma^{\tau_{2k-1}}0=0,
\end{align*}
where we have used the fact that by construction  $\theta^c_t\equiv 0$ whenever $t=\tau_1,\tau_2,\ldots$.


We now note that it is easy to see that $\hat{v}^{\boldsymbol{\pi},\mathfrak{g}}_c(s_0,I_0)$ is bounded above, indeed using the above we have that

\begin{align}
\hat{v}^{\boldsymbol{\pi},\mathfrak{g}}_c(s_0,I_0)&=\mathbb{E}_{\boldsymbol{\pi},g}\left[ \sum_{t=0}^\infty \gamma^t\left(\hat{R} -\sum_{k\geq 1} \delta^t_{\tau_{2k-1}}
+L_n(s_t)\right)\right]
\\&=\mathbb{E}_{\boldsymbol{\pi},g}\left[ \sum_{t=0}^\infty \gamma^t\left({R} -\sum_{k\geq 1} \delta^t_{\tau_{2k-1}}
+L_n(s_t)\right)+\sum_{t=0}^\infty \gamma^tFI_t\right]
\\&\leq \mathbb{E}_{\boldsymbol{\pi},g}\left[ \sum_{t=0}^\infty \gamma^t\left({R} +L_n(s_t)\right)\right]
\\&\leq \left|\mathbb{E}_{\boldsymbol{\pi},g}\left[ \sum_{t=0}^\infty \gamma^t\left({R}+L_n(s_t)\right)\right]\right|
\\&\leq \mathbb{E}_{\boldsymbol{\pi},g}\left[ \sum_{t=0}^\infty \gamma^t\left\|{R} +
L_n\right\|\right]
\\&\leq  \sum_{t=0}^\infty \gamma^t\left(\left\|{R}\right\| +\left\|
L_n\right\|\right)
\\&=\frac{1}{1-\gamma}\left(\left\|{R}\right\| +\left\|
L\right\|\right),
\end{align}
using the triangle inequality, the definition of $\hat{R}$ and the (upper-)boundedness of $L$ and $R$ (Assumption 5).
We now note that by the dominated convergence theorem we have that $\forall (s_0,I_0)\in\mathcal{S}\times\{0,1\}$
\begin{align}
&\underset{n\to\infty}{\lim}\; \hat{v}^{\boldsymbol{\pi},\mathfrak{g}}_c(s_0,I_0)  = \underset{n\to\infty}{\lim}\; \mathbb{E}_{\boldsymbol{\pi},g}\left[ \sum_{t=0}^\infty \gamma^t\left(\hat{R} -\sum_{k\geq 1} \delta^t_{\tau_{2k-1}}
+L_n(s_t)\right)\right]
\\&=\mathbb{E}_{\boldsymbol{\pi},g}\underset{n\to\infty}{\lim}\;\left[ \sum_{t=0}^\infty \gamma^t\left(\hat{R} -\sum_{k\geq 1} \delta^t_{\tau_{2k-1}}
+L_n(s_t)\right)\right]
\\&=\mathbb{E}_{\boldsymbol{\pi},g}\left[ \sum_{t=0}^\infty \gamma^t\left(\hat{R} -\sum_{k\geq 1} \delta^t_{\tau_{2k-1}}\right)\right]
\\&=\mathbb{E}_{\boldsymbol{\pi},g}\left[ \sum_{t=0}^\infty \gamma^t\left(R -\sum_{k\geq 1} \delta^t_{\tau_{2k-1}}\right)\right]=-\frac{K}{1-\gamma}+v^{\pi}(s_0),
\end{align}
using Assumption 6 in the last step, after which we deduce (i).

To deduce (ii) we simply note that $\hat{v}^{\boldsymbol{\pi},\mathfrak{g}}_c(s_0,I_0)$ and $v^{\pi}(s_0)$ differ by only a constant and hence share the same optimisation.

\end{proof}




\section*{Proof of Theorem \ref{theorem:existence_2}}
\begin{proof}
Theorem \ref{theorem:existence_2} is proved by firstly showing that when the players jointly maximise the same objective there exists a fixed point equilibrium of the game when all players use Markov policies and {\fontfamily{cmss}\selectfont Generator} uses switching control. The proof then proceeds by showing that the MG $\mathcal{G}$ admits a dual representation as an MG in which jointly maximise the same objective which has a stable point that can be computed by solving an MDP. Thereafter, we use both results to prove the existence of a fixed point for the game as a limit point of a sequence generated by successively applying the Bellman operator to a test function.  

Therefore, the scheme of the proof is summarised with the following steps:
\begin{itemize}
    \item[\textbf{I)}] Prove that the solution to Markov Team games (that is games in which both players maximise \textit{identical objectives}) in which one of the players uses switching control is the limit point of a sequence of Bellman operators (acting on some test function).
    \item[\textbf{II)}] Prove that for the MG $\mathcal{G}$ that is there exists a function $B^{\boldsymbol{\pi},\mathfrak{g}}:\mathcal{S}\times \{0,1\}\to \mathbb{R}$ such that\footnote{This property is analogous to  the condition in Markov potential games \citep{macua2018learning,mguni2021learning}} $
 v^{\boldsymbol{\pi},\mathfrak{g}}(z)-v^{\boldsymbol{\pi'},\mathfrak{g}}(z)
=B^{\boldsymbol{\pi},\mathfrak{g}}(z)-B^{\boldsymbol{\pi'},\mathfrak{g}}(z),\;\;\forall z\equiv (s,I_0)\in\mathcal{S}\times \{0,1\},\; \forall \mathfrak{g}$,
and
$
 \hat{v}^{\boldsymbol{\pi},\mathfrak{g}}_c(z)-\hat{v}^{\boldsymbol{\pi},\mathfrak{g}'}_c(z)
=B^{\boldsymbol{\pi},\mathfrak{g}}(z)-B^{\boldsymbol{\pi'},\mathfrak{g}}(z),\;\;\forall z\equiv (s,I_0)\in\mathcal{S}\times \{0,1\}, \forall \boldsymbol{\pi}\in\boldsymbol{\Pi}$,
    \item[\textbf{III)}] Prove that the MG $\mathcal{G}$ has a dual representation as a \textit{Markov Team Game} which admits a representation as an MDP.
\end{itemize}

\subsection*{Proof of Part \textbf{I}}

Our first result proves that the operator  $T$ is a contraction operator. First let us recall that the \textit{switching time} $\tau_k$ is defined recursively $\tau_k=\inf\{t>\tau_{k-1}|s_t\in A,\tau_k\in\mathcal{F}_t\}$ where $A=\{s\in \mathcal{S},m\in M|\mathfrak{g}_c(m|s_t)>0\}$.
To this end, we show that the following bounds holds:
\begin{lemma}\label{lemma:bellman_contraction}
The Bellman operator $T$ is a contraction, that is the following bound holds:
\begin{align*}
&\left\|T\psi-T\psi'\right\|\leq \gamma\left\|\psi-\psi'\right\|.
\end{align*}
\end{lemma}

\begin{proof}
Recall we define the Bellman operator $T_\psi$ of $\mathcal{G}$ acting on a function $\Lambda:\mathcal{S}\times\mathbb{N}\to\mathbb{R}$ by
\begin{align}
T_\psi \Lambda(s_{\tau_k},I(\tau_k)):=\max\left\{\mathcal{M}^{\boldsymbol{\pi},g}\Lambda(s_{\tau_k},I(\tau_k)),\left[ \psi(s_{\tau_k},\boldsymbol{a})+\gamma\underset{\boldsymbol{a}\in\boldsymbol{\mathcal{A}}}{\max}\;\sum_{s'\in\mathcal{S}}P(s';\boldsymbol{a},s_{\tau_k})\Lambda(s',I(\tau_k))\right]\right\}\label{bellman_proof_start}
\end{align}

In what follows and for the remainder of the script, we employ the following shorthands:
\begin{align*}
&\mathcal{P}^{\boldsymbol{a}}_{ss'}=:\sum_{s'\in\mathcal{S}}P(s';\boldsymbol{a},s), \quad\mathcal{P}^{\boldsymbol{\pi}}_{ss'}=:\sum_{\boldsymbol{a}\in\boldsymbol{\mathcal{A}}}\boldsymbol{\pi}(\boldsymbol{a}|s)\mathcal{P}^{\boldsymbol{a}}_{ss'}, \quad \mathcal{R}^{\boldsymbol{\pi}}(z_{t}):=\sum_{\boldsymbol{a}_t\in\boldsymbol{\mathcal{A}}}\boldsymbol{\pi}(\boldsymbol{a}_t|s)\hat{R}(z_t,\boldsymbol{a}_t,\theta_t,\theta_{t-1})
\end{align*}

To prove that $T$ is a contraction, we consider the three cases produced by \eqref{bellman_proof_start}, that is to say we prove the following statements:

i) $\qquad\qquad
\left| \Theta(z_t,\boldsymbol{a},\theta^c_t,\theta^c_{t-1})+\gamma\underset{\boldsymbol{a}\in\boldsymbol{\mathcal{A}}}{\max}\;\mathcal{P}^{\boldsymbol{a}}_{s's_t}\psi(s',\cdot)-\left( \Theta(z_t,\boldsymbol{a},\theta^c_t,\theta^c_{t-1})+\gamma\underset{\boldsymbol{a}\in\boldsymbol{\mathcal{A}}}{\max}\;\mathcal{P}^{\boldsymbol{a}}_{s's_t}\psi'(s',\cdot)\right)\right|\leq \gamma\left\|\psi-\psi'\right\|$

ii) $\qquad\qquad
\left\|\mathcal{M}^{\boldsymbol{\pi},g}\psi-\mathcal{M}^{\boldsymbol{\pi},g}\psi'\right\|\leq    \gamma\left\|\psi-\psi'\right\|,\qquad \qquad$
  (and hence $\mathcal{M}$ is a contraction).

iii) $\qquad\qquad
    \left\|\mathcal{M}^{\boldsymbol{\pi},g}\psi-\left[ \Theta(\cdot,\boldsymbol{a})+\gamma\underset{\boldsymbol{a}\in\boldsymbol{\mathcal{A}}}{\max}\;\mathcal{P}^{\boldsymbol{a}}\psi'\right]\right\|\leq \gamma\left\|\psi-\psi'\right\|.
$
where $z_t\equiv (s_t,I_t)\in\mathcal{S}\times \{0,1\}$.

We begin by proving i).

Indeed, for any $\boldsymbol{a}\in\boldsymbol{\mathcal{A}}$ and $\forall z_t\in\mathcal{S}\times\{0,1\}, \forall \theta_t,\theta_{t-1}\in \Theta, \forall s'\in\mathcal{S}$ we have that 
\begin{align*}
&\left| \Theta(z_t,\boldsymbol{a},\theta^c_t,\theta^c_{t-1})+\gamma\mathcal{P}^\pi_{s's_t}\psi(s',\cdot)-\left[ \Theta(z_t,\boldsymbol{a},\theta^c_t,\theta^c_{t-1})+\gamma\underset{\boldsymbol{a}\in\boldsymbol{\mathcal{A}}}{\max}\;\;\mathcal{P}^{\boldsymbol{a}}_{s's_t}\psi'(s',\cdot)\right]\right|
\\&\leq \underset{\boldsymbol{a}\in\boldsymbol{\mathcal{A}}}{\max}\;\left|\gamma\mathcal{P}^{\boldsymbol{a}}_{s's_t}\psi(s',\cdot)-\gamma\mathcal{P}^{\boldsymbol{a}}_{s's_t}\psi'(s',\cdot)\right|
\\&\leq \gamma\left\|P\psi-P\psi'\right\|
\\&\leq \gamma\left\|\psi-\psi'\right\|,
\end{align*}
again using the fact that $P$ is non-expansive and Lemma \ref{max_lemma}.

We now prove ii).

For any $\tau\in\mathcal{F}$, define by $\tau'=\inf\{t>\tau|s_t\in A,\tau\in\mathcal{F}_t\}$. Now using the definition of $\mathcal{M}$ we have that for any $s_\tau\in\mathcal{S}$
\begin{align*}
&\left|(\mathcal{M}^{\boldsymbol{\pi},g}\psi-\mathcal{M}^{\boldsymbol{\pi},g}\psi')(s_{\tau},I(\tau))\right|
\\&\leq \underset{\boldsymbol{a}_\tau,\theta^c_\tau,\theta^c_{\tau-1}\in \boldsymbol{\mathcal{A}}\times \Theta^2}{\max}    \Bigg|\Theta(z_\tau,\boldsymbol{a}_\tau,\theta^c_\tau,\theta^c_{\tau-1})-\delta^{\tau}_t+\gamma\mathcal{P}^{\boldsymbol{\pi}}_{s's_\tau}\mathcal{P}^{\boldsymbol{a}}\psi(s_{\tau},I(\tau'))
\\&\qquad\qquad-\left(\Theta(z_\tau,\boldsymbol{a}_\tau,\theta^c_\tau,\theta^c_{\tau-1})-\delta^{\tau}_t+\gamma\mathcal{P}^{\boldsymbol{\pi}}_{s's_\tau}\mathcal{P}^{\boldsymbol{a}}\psi'(s_{\tau},I(\tau'))\right)\Bigg| 
\\&= \gamma\left|\mathcal{P}^{\boldsymbol{\pi}}_{s's_\tau}\mathcal{P}^{\boldsymbol{a}}\psi(s_{\tau},I(\tau'))-\mathcal{P}^{\boldsymbol{\pi}}_{s's_\tau}\mathcal{P}^{\boldsymbol{a}}\psi'(s_{\tau},I(\tau'))\right| 
\\&\leq \gamma\left\|P\psi-P\psi'\right\|
\\&\leq \gamma\left\|\psi-\psi'\right\|,
\end{align*}
using the fact that $P$ is non-expansive. The result can then be deduced easily by applying max on both sides.

We now prove iii). We split the proof of the statement into two cases:

\textbf{Case 1:} 
\begin{align}\mathcal{M}^{\boldsymbol{\pi},g}\psi(s_{\tau},I(\tau))-\left(\Theta(z_\tau,\boldsymbol{a}_\tau,\theta^c_{\tau},\theta^c_{\tau-1})+\gamma\underset{\boldsymbol{a}\in\boldsymbol{\mathcal{A}}}{\max}\;\mathcal{P}^{\boldsymbol{a}}_{s's_\tau}\psi'(s',I(\tau))\right)<0.
\end{align}

We now observe the following:
\begin{align*}
&\mathcal{M}^{\boldsymbol{\pi},g}\psi(s_{\tau},I(\tau))-\Theta(z_\tau,\boldsymbol{a}_\tau,\theta^c_{\tau},\theta^c_{\tau-1})+\gamma\underset{\boldsymbol{a}\in\boldsymbol{\mathcal{A}}}{\max}\;\mathcal{P}^{\boldsymbol{a}}_{s's_\tau}\psi'(s',I(\tau))
\\&\leq\max\left\{\Theta(z_\tau,\boldsymbol{a}_\tau,\theta^c_{\tau},\theta^c_{\tau-1})+\gamma\mathcal{P}^{\boldsymbol{\pi}}_{s's_\tau}\mathcal{P}^{\boldsymbol{a}}\psi(s',I({\tau})),\mathcal{M}^{\boldsymbol{\pi},g}\psi(s_{\tau},I(\tau))\right\}
\\&\qquad-\Theta(z_\tau,\boldsymbol{a}_\tau,\theta^c_{\tau},\theta^c_{\tau-1})+\gamma\underset{\boldsymbol{a}\in\boldsymbol{\mathcal{A}}}{\max}\;\mathcal{P}^{\boldsymbol{a}}_{s's_\tau}\psi'(s',I(\tau))
\\&\leq \Bigg|\max\left\{\Theta(z_\tau,\boldsymbol{a}_\tau,\theta^c_{\tau},\theta^c_{\tau-1})+\gamma\mathcal{P}^{\boldsymbol{\pi}}_{s's_\tau}\mathcal{P}^{\boldsymbol{a}}\psi(s',I({\tau})),\mathcal{M}^{\boldsymbol{\pi},g}\psi(s_{\tau},I(\tau))\right\}
\\&\qquad-\max\left\{\Theta(z_\tau,\boldsymbol{a}_\tau,\theta^c_{\tau},\theta^c_{\tau-1})+\gamma\underset{\boldsymbol{a}\in\boldsymbol{\mathcal{A}}}{\max}\;\mathcal{P}^{\boldsymbol{a}}_{s's_\tau}\psi'(s',I({\tau})),\mathcal{M}^{\boldsymbol{\pi},g}\psi(s_{\tau},I(\tau))\right\}
\\&+\max\left\{\Theta(z_\tau,\boldsymbol{a}_\tau,\theta^c_{\tau},\theta^c_{\tau-1})+\gamma\underset{\boldsymbol{a}\in\boldsymbol{\mathcal{A}}}{\max}\;\mathcal{P}^{\boldsymbol{a}}_{s's_\tau}\psi'(s',I({\tau})),\mathcal{M}^{\boldsymbol{\pi},g}\psi(s_{\tau},I(\tau))\right\}
\\&\qquad-\Theta(z_\tau,\boldsymbol{a}_\tau,\theta^c_{\tau},\theta^c_{\tau-1})+\gamma\underset{\boldsymbol{a}\in\boldsymbol{\mathcal{A}}}{\max}\;\mathcal{P}^{\boldsymbol{a}}_{s's_\tau}\psi'(s',I(\tau))\Bigg|
\\&\leq \Bigg|\max\left\{\Theta(z_\tau,\boldsymbol{a}_\tau,\theta^c_{\tau},\theta^c_{\tau-1})+\gamma\underset{\boldsymbol{a}\in\boldsymbol{\mathcal{A}}}{\max}\;\mathcal{P}^{\boldsymbol{a}}_{s's_\tau}\psi(s',I({\tau})),\mathcal{M}^{\boldsymbol{\pi},g}\psi(s_{\tau},I(\tau))\right\}
\\&\qquad-\max\left\{\Theta(z_\tau,\boldsymbol{a}_\tau,\theta^c_{\tau},\theta^c_{\tau-1})+\gamma\underset{\boldsymbol{a}\in\boldsymbol{\mathcal{A}}}{\max}\;\mathcal{P}^{\boldsymbol{a}}_{s's_\tau}\psi'(s',I({\tau})),\mathcal{M}^{\boldsymbol{\pi},g}\psi(s_{\tau},I(\tau))\right\}\Bigg|
\\&\qquad+\Bigg|\max\left\{\Theta(z_\tau,\boldsymbol{a}_\tau,\theta^c_{\tau},\theta^c_{\tau-1})+\gamma\underset{\boldsymbol{a}\in\boldsymbol{\mathcal{A}}}{\max}\;\mathcal{P}^{\boldsymbol{a}}_{s's_\tau}\psi'(s',I({\tau})),\mathcal{M}^{\boldsymbol{\pi},g}\psi(s_{\tau},I(\tau))\right\}\\&\qquad\qquad-\Theta(z_\tau,\boldsymbol{a}_\tau,\theta^c_{\tau},\theta^c_{\tau-1})+\gamma\underset{\boldsymbol{a}\in\boldsymbol{\mathcal{A}}}{\max}\;\mathcal{P}^{\boldsymbol{a}}_{s's_\tau}\psi'(s',I(\tau))\Bigg|
\\&\leq \gamma\underset{a\in\mathcal{A}}{\max}\;\left|\mathcal{P}^{\boldsymbol{\pi}}_{s's_\tau}\mathcal{P}^{\boldsymbol{a}}\psi(s',I(\tau))-\mathcal{P}^{\boldsymbol{\pi}}_{s's_\tau}\mathcal{P}^{\boldsymbol{a}}\psi'(s',I(\tau))\right|
\\&\qquad+\left|\max\left\{0,\mathcal{M}^{\boldsymbol{\pi},g}\psi(s_{\tau},I(\tau))-\left(\Theta(z_\tau,\boldsymbol{a}_\tau,\theta^c_{\tau},\theta^c_{\tau-1})+\gamma\underset{\boldsymbol{a}\in\boldsymbol{\mathcal{A}}}{\max}\;\mathcal{P}^{\boldsymbol{a}}_{s's_\tau}\psi'(s',I(\tau))\right)\right\}\right|
\\&\leq \gamma\left\|P\psi-P\psi'\right\|
\\&\leq \gamma\|\psi-\psi'\|,
\end{align*}
where we have used the fact that for any scalars $a,b,c$ we have that $
    \left|\max\{a,b\}-\max\{b,c\}\right|\leq \left|a-c\right|$ and the non-expansiveness of $P$.

\textbf{Case 2: }
\begin{align*}\mathcal{M}^{\boldsymbol{\pi},g}\psi(s_{\tau},I(\tau))-\left(\Theta(z_\tau,\boldsymbol{a}_\tau,\theta^c_{\tau},\theta^c_{\tau-1})+\gamma\underset{\boldsymbol{a}\in\boldsymbol{\mathcal{A}}}{\max}\;\mathcal{P}^{\boldsymbol{a}}_{s's_\tau}\psi'(s',I(\tau))\right)\geq 0.
\end{align*}

\begin{align*}
&\mathcal{M}^{\boldsymbol{\pi},g}\psi(s_{\tau},I(\tau))-\left(\Theta(z_\tau,\boldsymbol{a}_\tau,\theta^c_{\tau},\theta^c_{\tau-1})+\gamma\underset{\boldsymbol{a}\in\boldsymbol{\mathcal{A}}}{\max}\;\mathcal{P}^{\boldsymbol{a}}_{s's_\tau}\psi'(s',I(\tau))\right)
\\&\leq \mathcal{M}^{\boldsymbol{\pi},g}\psi(s_{\tau},I(\tau))-\left(\Theta(z_\tau,\boldsymbol{a}_\tau,\theta^c_{\tau},\theta^c_{\tau-1})+\gamma\underset{\boldsymbol{a}\in\boldsymbol{\mathcal{A}}}{\max}\;\mathcal{P}^{\boldsymbol{a}}_{s's_\tau}\psi'(s',I(\tau))\right)+\delta^{\tau}_t
\\&\leq \Theta(z_\tau,\boldsymbol{a}_\tau,\theta^c_{\tau},\theta^c_{\tau-1})-\delta^{\tau}_t+\gamma\mathcal{P}^{\boldsymbol{\pi}}_{s's_\tau}\mathcal{P}^{\boldsymbol{a}}\psi(s',I(\tau'))
\\&\qquad\qquad\qquad\qquad\quad-\left(\Theta(z_\tau,\boldsymbol{a}_\tau,\theta^c_{\tau},\theta^c_{\tau-1})-\delta^{\tau}_t+\gamma\underset{\boldsymbol{a}\in\boldsymbol{\mathcal{A}}}{\max}\;\mathcal{P}^{\boldsymbol{a}}_{s's_\tau}\psi'(s',I(\tau))\right)
\\&\leq \gamma\underset{\boldsymbol{a}\in\boldsymbol{\mathcal{A}}}{\max}\;\left|\mathcal{P}^{\boldsymbol{\pi}}_{s's_\tau}\mathcal{P}^{\boldsymbol{a}}\left(\psi(s',I(\tau'))-\psi'(s',I(\tau))\right)\right|
\\&\leq \gamma\left|\psi(s',I(\tau'))-\psi'(s',I(\tau))\right|
\\&\leq \gamma\left\|\psi-\psi'\right\|,
\end{align*}
again using the fact that $P$ is non-expansive. Hence we have succeeded in showing that for any $\Lambda\in L_2$ we have that
\begin{align}
    \left\|\mathcal{M}^{\boldsymbol{\pi},g}\Lambda-\underset{\boldsymbol{a}\in\boldsymbol{\mathcal{A}}}{\max}\;\left[ \psi(\cdot,a)+\gamma\mathcal{P}^{\boldsymbol{a}}\Lambda'\right]\right\|\leq \gamma\left\|\Lambda-\Lambda'\right\|.\label{off_M_bound_gen}
\end{align}
Gathering the results of the three cases gives the desired result. 
\end{proof}
\subsection*{Proof of Part \textbf{II}}

To prove Part \textbf{II}, we prove the following result:
\begin{proposition}\label{dpg_proposition}
For any ${\pi}\in{\Pi}$ and for any {\fontfamily{cmss}\selectfont Generator} policy $\mathfrak{g}$, there exists a function $B^{\boldsymbol{\pi},\mathfrak{g}}:\mathcal{S}\times \{0,1\}\to \mathbb{R}$ such that
\begin{align}
 v^{\boldsymbol{\pi},\mathfrak{g}}_i-v^{{\boldsymbol{\pi'}},\mathfrak{g}}_i
=B^{\boldsymbol{\pi},\mathfrak{g}}(z)-B^{\boldsymbol{\pi'},\mathfrak{g}}(z),\;\;\forall z\equiv (s,I_0)\in\mathcal{S}\times \{0,1\}\label{potential_relation_proof}
\end{align}
where in particular the function $B$ is given by:
\begin{align}
B^{\boldsymbol{\pi},\mathfrak{g}}(s_0,I_0) =\mathbb{E}_{\boldsymbol{\pi},\mathfrak{g}}\left[ \sum_{t=0}^\infty \gamma^tR \right],\end{align}
for any $(s_0,I_0)\in\mathcal{S}\times \{0,1\}$.
\end{proposition}
\begin{proof}
Note that by the deduction of (ii) in Prop \ref{invariance_prop}, we may consider the following quantity for the {\fontfamily{cmss}\selectfont Generator} expected return:
\begin{align}
    \hat{v}^{\boldsymbol{\pi},\mathfrak{g}}_c(s_0,I_0)=\mathbb{E}_{\boldsymbol{\pi},\mathfrak{g}}\left[ \sum_{t=0}^\infty \gamma^t\left(R -\sum_{k\geq 1} \delta^t_{\tau_{2k-1}}\right)\right].
\end{align}

Therefore, we immediately observe that 
\begin{align}
    \hat{v}^{\boldsymbol{\pi},\mathfrak{g}}_c(s_0,I_0)=B^{\boldsymbol{\pi},\mathfrak{g}}(s_0,I_0)-K, \;\; \forall (s_0,I_0)\in\mathcal{S}\times\{0,1\}.
\end{align}
We therefore immediately deduce that for any two {\fontfamily{cmss}\selectfont Generator} policies $\mathfrak{g}$ and $\mathfrak{g}'$ the following expression holds $\forall (s_0,I_0)\in\mathcal{S}\times\{0,1\}$:
\begin{align}
    \hat{v}^{\boldsymbol{\pi},\mathfrak{g}}_c(s_0,I_0)-\hat{v}^{\boldsymbol{\pi},\mathfrak{g}'}_c(s_0,I_0)=B^{\boldsymbol{\pi},\mathfrak{g}}(s_0,I_0)-B^{{\pi},\mathfrak{g}'}(s_0,I_0).
\end{align}

Our aim now is to show that the following expression holds $\forall (s_0,I_0)\in\mathcal{S}\times\{0,1\}$:
\begin{align}\nonumber
 \hat{v}^{\boldsymbol{\pi},\mathfrak{g}}_c(I_{0},s_{0})-\hat{v}^{\boldsymbol{\pi'},\mathfrak{g}}_c(I_{0},s_{0})=B^{\boldsymbol{\pi},\mathfrak{g}}(I_{0},s_{0})-B^{\boldsymbol{\pi'},\mathfrak{g}}(I_{0},s_{0}),
\end{align}
This is manifest from the construction of $B$.
\end{proof}
\subsection*{Proof of Part \textbf{III}}
To prove Part \textbf{III}, we firstly define precisely the notion of a stable point of the MG, $\mathcal{G}$:
\begin{definition}
A policy profile $\boldsymbol{\sigma^\star}=(g^\star,\pi^\star_i,\pi_{-i}^\star)\in\boldsymbol{\Pi}$ is a Markov perfect equilibrium (MPE) in Markov strategies if the following condition holds for any $i\in\mathcal{N}\times\{0\}$:
\begin{align}\label{MP_NE_condition}
&v_i^{(g^\star,\pi^\star_i,\pi^\star_{-i})}(z)\geq v_i^{g^\star,(\pi'_i,\pi^\star_{-i})}(z), \; \forall z\equiv (s_0,I_0)\in \mathcal{S}\times\{0,1\}, \;\forall \pi_i'\in\Pi_i.
\\&v_c^{(g^\star,\pi^\star_i,\pi^\star_{-i})}(z)\geq v_c^{g',(\pi_i,\pi^\star_{-i})}(z), \; \forall z\equiv (s_0,I_0)\in \mathcal{S}\times\{0,1\}, \;\forall g'.\end{align}
\end{definition}
The condition characterises strategic configurations which are stable points of the MG, $\mathcal{G}$. In particular, an MPE is achieved when at any state no agent can improve their expected cumulative rewards by unilaterally deviating from their current policy. We denote by $NE\{\mathcal{G}\}$ the set of MPE strategies for the MG, $\mathcal{G}$.

Next we prove that the set of maxima of the function $B$ are the MPE of the MG $\mathcal{G}$:

\begin{proposition}\label{reduction_prop}
The following implication holds: 
\begin{align}
\boldsymbol{\sigma}\in \underset{{g',\boldsymbol{\pi'}}\in\boldsymbol{\Pi}}{\arg\sup}\; B^{g',{\boldsymbol{\pi'}}}(s)\implies \boldsymbol{\sigma}\in NE\{\mathcal{G}\}.
\end{align}
where $B$ is the function in Prop. \ref{dpg_proposition}.
\end{proposition}
Prop. \ref{reduction_prop} indicates that the game has an equivalent representation in which all agents maximise the same function and thus  play a \textit{team game}.
\begin{proof}
We do the proof by contradiction. Let $\boldsymbol{\sigma}=(\pi_1,\ldots,\pi_N,g)\in \underset{\boldsymbol{\pi'}\in\boldsymbol{\Pi},g'}{\arg\sup}\; B^{\boldsymbol{\pi'},g'}(s)$ for any $s\in\mathcal{S}$. Let us now therefore assume that $\boldsymbol{\sigma}\notin NE\{\mathcal{G}\}$, hence there exists some other policy profile $\boldsymbol{\tilde{\sigma}}=(\pi_1,\ldots,\tilde{\pi}_i,\ldots,\pi_N,g)$ which contains at least one profitable deviation by one of the agents $i\in\mathcal{N}\times\{0,\}$. For now let us consider the case in which the profitable deviation is for a agent $i\in\mathcal{N}$  so that $\pi_i'\neq \pi_i$ for $i\in\mathcal{N}$ i.e. $v^{(\pi'_i,\pi_{-i}),g}_i(s)> v^{(\pi_i,\pi_{-i}),g}_i(s)$ (using the preservation of signs of integration). Prop. \ref{dpg_proposition} however implies that $B^{(\pi'_i,\pi_{-i}),g}(s)-B^{(\pi_i,\pi_{-i}),g}(s)>0$ which is a contradiction since $\boldsymbol{\sigma}=(\pi_i,\pi_{-i},g)$ is a maximum of $B$.  The proof can be straightforwardly adapted to cover the case in which the deviating agent is the  {\fontfamily{cmss}\selectfont Generator} after which we deduce the desired result.
\end{proof}
The last result completes the proof of Theorem \ref{theorem:existence_2}. 
\end{proof}
\section*{Proof of Proposition \ref{prop:switching_times}}
\begin{proof}
The proof is given by establishing a contradiction. Therefore suppose that $\mathcal{M}^{\boldsymbol{\pi},g}\psi(s_{\tau_k},I(\tau_k))\leq \psi(s_{\tau_k},I(\tau_k))$ and suppose that the switching time $\tau'_1>\tau_1$ is an optimal switching time. Construct the {\fontfamily{cmss}\selectfont Generator} $g'$ and $\tilde{g}$ policy switching times by $(\tau'_0,\tau'_1,\ldots,)$ and $g'^2$ policy by $(\tau'_0,\tau_1,\ldots)$ respectively.  Define by $l=\inf\{t>0;\mathcal{M}^{\boldsymbol{\pi},g}\psi(s_{t},I_0)= \psi(s_{t},I_0)\}$ and $m=\sup\{t;t<\tau'_1\}$.
By construction we have that
\begin{align*}
& \quad v^{\boldsymbol{\pi},g'}_c(s,I_0)
\\&=\mathbb{E}\left[R(s_{0},\boldsymbol{a}_{0})+\mathbb{E}\left[\ldots+\gamma^{l-1}\mathbb{E}\left[R(s_{\tau_1-1},\boldsymbol{a}_{\tau_1-1})+\ldots+\gamma^{m-l-1}\mathbb{E}\left[ R(s_{\tau'_1-1},\boldsymbol{a}_{\tau'_1-1})+\gamma\mathcal{M}^{\boldsymbol{\pi},\mathfrak{g}}v^{\boldsymbol{\pi},g'}_c(s',I(\tau'_{1}))\right]\right]\right]\right]
\\&<\mathbb{E}\left[R(s_{0},\boldsymbol{a}_{0})+\mathbb{E}\left[\ldots+\gamma^{l-1}\mathbb{E}\left[ R(s_{\tau_1-1},\boldsymbol{a}_{\tau_1-1})+\gamma\mathcal{M}^{\boldsymbol{\pi},\tilde{g}}v^{\boldsymbol{\pi},g'}_c(s_{\tau_1},I(\tau_1))\right]\right]\right]
\end{align*}
We now use the following observation $\mathbb{E}\left[ R(s_{\tau_1-1},\boldsymbol{a}_{\tau_1-1})+\gamma\mathcal{M}^{\boldsymbol{\pi},\tilde{g}}v^{\boldsymbol{\pi},g'}_c(s_{\tau_1},I(\tau_1))\right]\\\ \text{\hspace{30 mm}}\leq \max\left\{\mathcal{M}^{\boldsymbol{\pi},\tilde{g}}v^{\boldsymbol{\pi},g'}_c(s_{\tau_1},I(\tau_1)),\underset{a_{\tau_1}\in\mathcal{A}}{\max}\;\left[ R(s_{\tau_{k}},\boldsymbol{a}_{\tau_{k}})+\gamma\sum_{s'\in\mathcal{S}}P(s';\boldsymbol{a}_{\tau_1},s_{\tau_1})v^{\boldsymbol{\pi},g}_c(s',I(\tau_1))\right]\right\}$.

Using this we deduce that
\begin{align*}
&v^{\boldsymbol{\pi},g'}_2(s,I_0)\leq\mathbb{E}\Bigg[R(s_{0},\boldsymbol{a}_{0})+\mathbb{E}\Bigg[\ldots
\\&+\gamma^{l-1}\mathbb{E}\left[ R(s_{\tau_1-1},\boldsymbol{a}_{\tau_1-1})+\gamma\max\left\{\mathcal{M}^{\boldsymbol{\pi},\tilde{g}}v^{\boldsymbol{\pi},g'}_c(s_{\tau_1},I(\tau_1)),\underset{a_{\tau_1}\in\mathcal{A}}{\max}\;\left[ R(s_{\tau_{k}},\boldsymbol{a}_{\tau_{k}})+\gamma\sum_{s'\in\mathcal{S}}P(s';\boldsymbol{a}_{\tau_1},s_{\tau_1})v^{\boldsymbol{\pi},g}_c(s',I(\tau_1))\right]\right\}\right]\Bigg]\Bigg]
\\&=\mathbb{E}\left[R(s_{0},\boldsymbol{a}_{0})+\mathbb{E}\left[\ldots+\gamma^{l-1}\mathbb{E}\left[ R(s_{\tau_1-1},\boldsymbol{a}_{\tau_1-1})+\gamma\left[T v^{\boldsymbol{\pi},\tilde{g}}_c\right](s_{\tau_1},I(\tau_1))\right]\right]\right]=v^{\boldsymbol{\pi},\tilde{g}}_c(s,I_0))
\end{align*}
where the first inequality is true by assumption on $\mathcal{M}$. This is a contradiction since $g'$ is an optimal policy for the {\fontfamily{cmss}\selectfont Generator}. Using analogous reasoning, we deduce the same result for $\tau'_k<\tau_k$ after which deduce the result. Moreover, by invoking the same reasoning, we can conclude that it must be the case that $(\tau_0,\tau_1,\ldots,\tau_{k-1},\tau_k,\tau_{k+1},\ldots,)$ are the optimal switching times.

\end{proof}

\section*{Proof of Theorem \ref{NE_improve_prop}}
\begin{proof}
The proof which is done by contradiction follows from the definition of $v_c$.
Denote by $v^{\boldsymbol{\pi},g\equiv \boldsymbol{0}}_i$  value function an agent $i\in\mathcal{N}$ \textit{excluding the {\fontfamily{cmss}\selectfont Generator}} and its intrinsic-reward function. 
Indeed, 
let $(\boldsymbol{\hat{\pi}},\hat{g})$ be the policy profile induced by the Nash equilibrium policy profile and assume that the intrinsic-reward $F$ leads to a decrease in payoff for agent $i$. Then by construction 
$    v^{\boldsymbol{\pi},g}(s)< v^{\boldsymbol{\pi},g\equiv \boldsymbol{0}}(s)
$
which is a contradiction since $(\boldsymbol{\hat{\pi}},\hat{g})$ is an MPE profile.
\end{proof}

\section*{Proof of Theorem  \ref{primal_convergence_theorem}}
To prove the theorem, we make use of the following result:
\begin{theorem}[Theorem 1, pg 4 in \cite{jaakkola1994convergence}]
Let $\Xi_t(s)$ be a random process that takes values in $\mathbb{R}^n$ and given by the following:
\begin{align}
    \Xi_{t+1}(s)=\left(1-\alpha_t(s)\right)\Xi_{t}(s)\alpha_t(s)L_t(s),
\end{align}
then $\Xi_t(s)$ converges to $0$ with probability $1$ under the following conditions:
\begin{itemize}
\item[i)] $0\leq \alpha_t\leq 1, \sum_t\alpha_t=\infty$ and $\sum_t\alpha_t<\infty$
\item[ii)] $\|\mathbb{E}[L_t|\mathcal{F}_t]\|\leq \gamma \|\Xi_t\|$, with $\gamma <1$;
\item[iii)] ${\rm Var}\left[L_t|\mathcal{F}_t\right]\leq c(1+\|\Xi_t\|^2)$ for some $c>0$.
\end{itemize}
\end{theorem}
\begin{proof}
To prove the result, we show (i) - (iii) hold. Condition (i) holds by choice of learning rate. It therefore remains to prove (ii) - (iii). We first prove (ii). For this, we consider our variant of the Q-learning update rule:
\begin{align*}
Q_{t+1}(s_t,I_t,\boldsymbol{a}_t)=Q_{t}&(s_t,I_t,\boldsymbol{a}_t)
\\&+\alpha_t(s_t,I_t,\boldsymbol{a}_t)\left[\max\left\{\mathcal{M}^{\boldsymbol{\pi},g}Q(s_{\tau_k},I_{\tau_k},\boldsymbol{a}), \phi(s_{\tau_k},\boldsymbol{a})+\gamma\underset{a'\in\mathcal{A}}{\max}\;Q(s',I_{\tau_k},\boldsymbol{a'})\right\}-Q_{t}(s_t,I_t,\boldsymbol{a}_t)\right].
\end{align*}
After subtracting $Q^\star(s_t,I_t,\boldsymbol{a}_t)$ from both sides and some manipulation we obtain that:
\begin{align*}
&\Xi_{t+1}(s_t,I_t,\boldsymbol{a}_t)
\\&=(1-\alpha_t(s_t,I_t,\boldsymbol{a}_t))\Xi_{t}(s_t,I_t,\boldsymbol{a}_t)
\\&\qquad\qquad\qquad\qquad\;\;+\alpha_t(s_t,I_t,\boldsymbol{a}_t))\left[\max\left\{\mathcal{M}^{\boldsymbol{\pi},g}Q(s_{\tau_k},I_{\tau_k},\boldsymbol{a}), \phi(s_{\tau_k},\boldsymbol{a})+\gamma\underset{a'\in\mathcal{A}}{\max}\;Q(s',I_{\tau_k},\boldsymbol{a'})\right\}-Q^\star(s_t,I_t,\boldsymbol{a}_t)\right],  \end{align*}
where $\Xi_{t}(s_t,I_t,\boldsymbol{a}_t):=Q_t(s_t,I_t,\boldsymbol{a}_t)-Q^\star(s_t,I_t,\boldsymbol{a}_t)$.

Let us now define by 
\begin{align*}
L_t(s_{\tau_k},I_{\tau_k},\boldsymbol{a}):=\max\left\{\mathcal{M}^{\boldsymbol{\pi},g}Q(s_{\tau_k},I_{\tau_k},\boldsymbol{a}), \phi(s_{\tau_k},\boldsymbol{a})+\gamma\underset{a'\in\mathcal{A}}{\max}\;Q(s',I_{\tau_k},\boldsymbol{a'})\right\}-Q^\star(s_t,I_t,a).
\end{align*}
Then
\begin{align}
\Xi_{t+1}(s_t,I_t,\boldsymbol{a}_t)=(1-\alpha_t(s_t,I_t,\boldsymbol{a}_t))\Xi_{t}(s_t,I_t,\boldsymbol{a}_t)+\alpha_t(s_t,I_t,\boldsymbol{a}_t))\left[L_t(s_{\tau_k},a)\right].   
\end{align}

We now observe that
\begin{align}\nonumber
\mathbb{E}\left[L_t(s_{\tau_k},I_{\tau_k},\boldsymbol{a})|\mathcal{F}_t\right]&=\sum_{s'\in\mathcal{S}}P(s';a,s_{\tau_k})\max\left\{\mathcal{M}^{\boldsymbol{\pi},g}Q(s_{\tau_k},I_{\tau_k},\boldsymbol{a}), \phi(s_{\tau_k},\boldsymbol{a})+\gamma\underset{a'\in\mathcal{A}}{\max}\;Q(s',I_{\tau_k},\boldsymbol{a'})\right\}-Q^\star(s_{\tau_k},a)
\\&= T_\phi Q_t(s,I_{\tau_k},\boldsymbol{a})-Q^\star(s,I_{\tau_k},\boldsymbol{a}). \label{expectation_L}
\end{align}
Now, using the fixed point property that implies $Q^\star=T_\phi Q^\star$, we find that
\begin{align}\nonumber
    \mathbb{E}\left[L_t(s_{\tau_k},I_{\tau_k},\boldsymbol{a})|\mathcal{F}_t\right]&=T_\phi Q_t(s,I_{\tau_k},\boldsymbol{a})-T_\phi Q^\star(s,I_{\tau_k},\boldsymbol{a})
    \\&\leq\left\|T_\phi Q_t-T_\phi Q^\star\right\|\nonumber
    \\&\leq \gamma\left\| Q_t- Q^\star\right\|_\infty=\gamma\left\|\Xi_t\right\|_\infty.
\end{align}
using the contraction property of $T$ established in Lemma \ref{lemma:bellman_contraction}. This proves (ii).

We now prove iii), that is
\begin{align}
    {\rm Var}\left[L_t|\mathcal{F}_t\right]\leq c(1+\|\Xi_t\|^2).
\end{align}
Now by \eqref{expectation_L} we have that
\begin{align*}
  {\rm Var}\left[L_t|\mathcal{F}_t\right]&= {\rm Var}\left[\max\left\{\mathcal{M}^{\boldsymbol{\pi},g}Q(s_{\tau_k},I_{\tau_k},\boldsymbol{a}), \phi(s_{\tau_k},\boldsymbol{a})+\gamma\underset{a'\in\mathcal{A}}{\max}\;Q(s',I_{\tau_k},\boldsymbol{a'})\right\}-Q^\star(s_t,I_t,a)\right]
  \\&= \mathbb{E}\Bigg[\Bigg(\max\left\{\mathcal{M}^{\boldsymbol{\pi},g}Q(s_{\tau_k},I_{\tau_k},\boldsymbol{a}), \phi(s_{\tau_k},\boldsymbol{a})+\gamma\underset{a'\in\mathcal{A}}{\max}\;Q(s',I_{\tau_k},\boldsymbol{a'})\right\}
  \\&\qquad\qquad\qquad\qquad\qquad\quad\quad\quad-Q^\star(s_t,I_t,a)-\left(T_\Phi Q_t(s,I_{\tau_k},\boldsymbol{a})-Q^\star(s,I_{\tau_k},\boldsymbol{a})\right)\Bigg)^2\Bigg]
      \\&= \mathbb{E}\left[\left(\max\left\{\mathcal{M}^{\boldsymbol{\pi},g}Q(s_{\tau_k},I_{\tau_k},\boldsymbol{a}), \phi(s_{\tau_k},\boldsymbol{a})+\gamma\underset{a'\in\mathcal{A}}{\max}\;Q(s',I_{\tau_k},\boldsymbol{a'})\right\}-T_\Phi Q_t(s,I_{\tau_k},\boldsymbol{a})\right)^2\right]
    \\&= {\rm Var}\left[\max\left\{\mathcal{M}^{\boldsymbol{\pi},g}Q(s_{\tau_k},I_{\tau_k},\boldsymbol{a}), \phi(s_{\tau_k},\boldsymbol{a})+\gamma\underset{a'\in\mathcal{A}}{\max}\;Q(s',I_{\tau_k},\boldsymbol{a'})\right\}-T_\Phi Q_t(s,I_{\tau_k},\boldsymbol{a}))^2\right]
    \\&\leq c(1+\|\Xi_t\|^2),
\end{align*}
for some $c>0$ where the last line follows due to the boundedness of $Q$ (which follows from Assumptions 2 and 4). This concludes the proof of the Theorem.
\end{proof}
\section*{Proof of Convergence with Function Approximation}
First let us recall the statement of the theorem:
\begin{customthm}{3}
LIGS converges to a limit point $r^\star$ which is the unique solution to the equation:
\begin{align}
\Pi \mathfrak{F} (\Phi r^\star)=\Phi r^\star, \qquad \text{a.e.}
\end{align}
where we recall that for any test function $\Lambda \in \mathcal{V}$, the operator $\mathfrak{F}$ is defined by $
    \mathfrak{F}\Lambda:=\Theta+\gamma P \max\{\mathcal{M}\Lambda,\Lambda\}$.

Moreover, $r^\star$ satisfies the following:
\begin{align}
    \left\|\Phi r^\star - Q^\star\right\|\leq c\left\|\Pi Q^\star-Q^\star\right\|.
\end{align}
\end{customthm}

The theorem is proven using a set of results that we now establish. To this end, we first wish to prove the following bound:    
\begin{lemma}
For any $Q\in\mathcal{V}$ we have that
\begin{align}
    \left\|\mathfrak{F}Q-Q'\right\|\leq \gamma\left\|Q-Q'\right\|,
\end{align}
so that the operator $\mathfrak{F}$ is a contraction.
\end{lemma}
\begin{proof}
Recall, for any test function $\psi$ , a projection operator $\Pi$ acting $\Lambda$ is defined by the following 
\begin{align*}
\Pi \Lambda:=\underset{\bar{\Lambda}\in\{\Phi r|r\in\mathbb{R}^p\}}{\arg\min}\left\|\bar{\Lambda}-\Lambda\right\|. 
\end{align*}
Now, we first note that in the proof of Lemma \ref{lemma:bellman_contraction}, we deduced that for any $\Lambda\in L_2$ we have that
\begin{align*}
    \left\|\mathcal{M}\Lambda-\left[ \psi(\cdot,a)+\gamma\underset{\boldsymbol{a}\in\boldsymbol{\mathcal{A}}}{\max}\;\mathcal{P}^{\boldsymbol{a}}\Lambda'\right]\right\|\leq \gamma\left\|\Lambda-\Lambda'\right\|,
\end{align*}
(c.f. Lemma \ref{lemma:bellman_contraction}). 

Setting $\Lambda=Q$ and $\psi=\Theta$, it can be straightforwardly deduced that for any $Q,\hat{Q}\in L_2$:
    $\left\|\mathcal{M}Q-\hat{Q}\right\|\leq \gamma\left\|Q-\hat{Q}\right\|$. Hence, using the contraction property of $\mathcal{M}$, we readily deduce the following bound:
\begin{align}\max\left\{\left\|\mathcal{M}Q-\hat{Q}\right\|,\left\|\mathcal{M}Q-\mathcal{M}\hat{Q}\right\|\right\}\leq \gamma\left\|Q-\hat{Q}\right\|,
\label{m_bound_q_twice}
\end{align}
    
We now observe that $\mathfrak{F}$ is a contraction. Indeed, since for any $Q,Q'\in L_2$ we have that:
%
%
%
\begin{align*}
\left\|\mathfrak{F}Q-\mathfrak{F}Q'\right\|&=\left\|\Theta+\gamma P \max\{\mathcal{M}Q,Q\}-\left(\Theta+\gamma P \max\{\mathcal{M}Q',Q'\}\right)\right\|
\\&=\gamma \left\|P \max\{\mathcal{M}Q,Q\}-P \max\{\mathcal{M}Q',Q'\}\right\|
\\&\leq\gamma \left\| \max\{\mathcal{M}Q,Q\}- \max\{\mathcal{M}Q',Q'\}\right\|
\\&\leq\gamma \left\| \max\{\mathcal{M}Q-\mathcal{M}Q',Q-\mathcal{M}Q',\mathcal{M}Q-Q',Q-Q'\}\right\|
\\&\leq\gamma \max\{\left\|\mathcal{M}Q-\mathcal{M}Q'\right\|,\left\|Q-\mathcal{M}Q'\right\|,\left\|\mathcal{M}Q-Q'\right\|,\left\|Q-Q'\right\|\}
\\&=\gamma\left\|Q-Q'\right\|,
\end{align*}
using \eqref{m_bound_q_twice} and again using the non-expansiveness of $P$.
\end{proof}
We next show that the following two bounds hold:
\begin{lemma}\label{projection_F_contraction_lemma}
For any $Q\in\mathcal{V}$ we have that
\begin{itemize}
    \item[i)] 
$\qquad\qquad
    \left\|\Pi \mathfrak{F}Q-\Pi \mathfrak{F}\bar{Q}\right\|\leq \gamma\left\|Q-\bar{Q}\right\|$,
    \item[ii)]$\qquad\qquad\left\|\Phi r^\star - Q^\star\right\|\leq \frac{1}{\sqrt{1-\gamma^2}}\left\|\Pi Q^\star - Q^\star\right\|$. 
\end{itemize}
\end{lemma}
\begin{proof}
The first result is straightforward since as $\Pi$ is a projection it is non-expansive and hence:
\begin{align*}
    \left\|\Pi \mathfrak{F}Q-\Pi \mathfrak{F}\bar{Q}\right\|\leq \left\| \mathfrak{F}Q-\mathfrak{F}\bar{Q}\right\|\leq \gamma \left\|Q-\bar{Q}\right\|,
\end{align*}
using the contraction property of $\mathfrak{F}$. This proves i). For ii), we note that by the orthogonality property of projections we have that $\left\langle\Phi r^\star - \Pi Q^\star,\Phi r^\star - \Pi Q^\star\right\rangle$, hence we observe that:
\begin{align*}
    \left\|\Phi r^\star - Q^\star\right\|^2&=\left\|\Phi r^\star - \Pi Q^\star\right\|^2+\left\|\Phi r^\star - \Pi Q^\star\right\|^2
\\&=\left\|\Pi \mathfrak{F}\Phi r^\star - \Pi Q^\star\right\|^2+\left\|\Phi r^\star - \Pi Q^\star\right\|^2
\\&\leq\left\|\mathfrak{F}\Phi r^\star -  Q^\star\right\|^2+\left\|\Phi r^\star - \Pi Q^\star\right\|^2
\\&=\left\|\mathfrak{F}\Phi r^\star -  \mathfrak{F}Q^\star\right\|^2+\left\|\Phi r^\star - \Pi Q^\star\right\|^2
\\&\leq\gamma^2\left\|\Phi r^\star -  Q^\star\right\|^2+\left\|\Phi r^\star - \Pi Q^\star\right\|^2,
\end{align*}
after which we readily deduce the desired result.
\end{proof}

\begin{lemma}
Define  the operator $H$ by the following: $
  HQ(z)=  \begin{cases}
			\mathcal{M}Q(z), & \text{if $\mathcal{M}Q(z)>\Phi r^\star,$}\\
            Q(z), & \text{otherwise},
		 \end{cases}$
\\and $\tilde{\mathfrak{F}}$ by: $
    \tilde{\mathfrak{F}}Q:=\Theta +\gamma PHQ$.

For any $Q,\bar{Q}\in L_2$ we have that
\begin{align}
    \left\|\tilde{\mathfrak{F}}Q-\tilde{\mathfrak{F}}\bar{Q}\right\|\leq \gamma \left\|Q-\bar{Q}\right\|
\end{align}
and hence $\tilde{\mathfrak{F}}$ is a contraction mapping.
\end{lemma}
\begin{proof}
Using \eqref{m_bound_q_twice}, we now observe that
\begin{align*}
    \left\|\tilde{\mathfrak{F}}Q-\tilde{\mathfrak{F}}\bar{Q}\right\|&=\left\|\Theta+\gamma PHQ -\left(\Theta+\gamma PH\bar{Q}\right)\right\|
\\&\leq \gamma\left\|HQ - H\bar{Q}\right\|
\\&\leq \gamma\left\|\max\left\{\mathcal{M}Q-\mathcal{M}\bar{Q},Q-\bar{Q},\mathcal{M}Q-\bar{Q},\mathcal{M}\bar{Q}-Q\right\}\right\|
\\&\leq \gamma\max\left\{\left\|\mathcal{M}Q-\mathcal{M}\bar{Q}\right\|,\left\|Q-\bar{Q}\right\|,\left\|\mathcal{M}Q-\bar{Q}\right\|,\left\|\mathcal{M}\bar{Q}-Q\right\|\right\}
\\&\leq \gamma\max\left\{\gamma\left\|Q-\bar{Q}\right\|,\left\|Q-\bar{Q}\right\|,\left\|\mathcal{M}Q-\bar{Q}\right\|,\left\|\mathcal{M}\bar{Q}-Q\right\|\right\}
\\&=\gamma\left\|Q-\bar{Q}\right\|,
\end{align*}
again using the non-expansive property of $P$.
\end{proof}
\begin{lemma}
Define by $\tilde{Q}:=\Theta+\gamma Pv^{\boldsymbol{\tilde{\pi}}}$ where
\begin{align}
    v^{\boldsymbol{\tilde{\pi}}}(z):= \Theta(s_{\tau_k},a)+\gamma\underset{a\in\mathcal{A}}{\max}\;\sum_{s'\in\mathcal{S}}P(s';a,s_{\tau_k})\Phi r^\star(s',I(\tau_k)), \label{v_tilde_definition}
\end{align}
then $\tilde{Q}$ is a fixed point of $\tilde{\mathfrak{F}}\tilde{Q}$, that is $\tilde{\mathfrak{F}}\tilde{Q}=\tilde{Q}$. 
\end{lemma}
\begin{proof}
We begin by observing that
\begin{align*}
H\tilde{Q}(z)&=H\left(\Theta(z)+\gamma Pv^{\boldsymbol{\tilde{\pi}}}\right)    
\\&= \begin{cases}
			\mathcal{M}Q(z), & \text{if $\mathcal{M}Q(z)>\Phi r^\star,$}\\
            Q(z), & \text{otherwise},
		 \end{cases}
\\&= \begin{cases}
			\mathcal{M}Q(z), & \text{if $\mathcal{M}Q(z)>\Phi r^\star,$}\\
            \Theta(z)+\gamma Pv^{\boldsymbol{\tilde{\pi}}}, & \text{otherwise},
		 \end{cases}
\\&=v^{\boldsymbol{\tilde{\pi}}}(z).
\end{align*}
Hence,
\begin{align}
    \tilde{\mathfrak{F}}\tilde{Q}=\Theta+\gamma PH\tilde{Q}=\Theta+\gamma Pv^{\boldsymbol{\tilde{\pi}}}=\tilde{Q}. 
\end{align}
which proves the result.
\end{proof}
\begin{lemma}\label{value_difference_Q_difference}
The following bound holds:
\begin{align}
    \mathbb{E}\left[v^{\boldsymbol{\hat{\pi}}}(z_0)\right]-\mathbb{E}\left[v^{\boldsymbol{\tilde{\pi}}}(z_0)\right]\leq 2\left[(1-\gamma)\sqrt{(1-\gamma^2)}\right]^{-1}\left\|\Pi Q^\star -Q^\star\right\|.
\label{F_tilde_fixed_point}\end{align}
\end{lemma}
\begin{proof}

By definitions of $v^{\boldsymbol{\hat{\pi}}}$ and $v^{\boldsymbol{\tilde{\pi}}}$ (c.f \eqref{v_tilde_definition}) and using Jensen's inequality and the stationarity property we have that,
\begin{align}\nonumber
    \mathbb{E}\left[v^{\boldsymbol{\hat{\pi}}}(z_0)\right]-\mathbb{E}\left[v^{\boldsymbol{\tilde{\pi}}}(z_0)\right]&=\mathbb{E}\left[Pv^{\boldsymbol{\hat{\pi}}}(z_0)\right]-\mathbb{E}\left[Pv^{\boldsymbol{\tilde{\pi}}}(z_0)\right]
    \\&\leq \left|\mathbb{E}\left[Pv^{\boldsymbol{\hat{\pi}}}(z_0)\right]-\mathbb{E}\left[Pv^{\boldsymbol{\tilde{\pi}}}(z_0)\right]\right|\nonumber
    \\&\leq \left\|Pv^{\boldsymbol{\hat{\pi}}}-Pv^{\boldsymbol{\tilde{\pi}}}\right\|. \label{v_approx_intermediate_bound_P}
\end{align}
Now recall that $\tilde{Q}:=\Theta+\gamma Pv^{\boldsymbol{\tilde{\pi}}}$ and $Q^\star:=\Theta+\gamma Pv^{\boldsymbol{\pi^\star}}$,  using these expressions in \eqref{v_approx_intermediate_bound_P} we find that 
\begin{align*}
    \mathbb{E}\left[v^{\boldsymbol{\hat{\pi}}}(z_0)\right]-\mathbb{E}\left[v^{\boldsymbol{\tilde{\pi}}}(z_0)\right]&\leq \frac{1}{\gamma}\left\|\tilde{Q}-Q^\star\right\|. \label{v_approx_q_approx_bound}
\end{align*}
Moreover, by the triangle inequality and using the fact that $\mathfrak{F}(\Phi r^\star)=\tilde{\mathfrak{F}}(\Phi r^\star)$ and that $\mathfrak{F}Q^\star=Q^\star$ and $\mathfrak{F}\tilde{Q}=\tilde{Q}$ (c.f. \eqref{F_tilde_fixed_point}) we have that
\begin{align*}
\left\|\tilde{Q}-Q^\star\right\|&\leq \left\|\tilde{Q}-\mathfrak{F}(\Phi r^\star)\right\|+\left\|Q^\star-\tilde{\mathfrak{F}}(\Phi r^\star)\right\|    
\\&\leq \gamma\left\|\tilde{Q}-\Phi r^\star\right\|+\gamma\left\|Q^\star-\Phi r^\star\right\| 
\\&\leq 2\gamma\left\|\tilde{Q}-\Phi r^\star\right\|+\gamma\left\|Q^\star-\tilde{Q}\right\|, 
\end{align*}
which gives the following bound:
\begin{align*}
\left\|\tilde{Q}-Q^\star\right\|&\leq 2\left(1-\gamma\right)^{-1}\left\|\tilde{Q}-\Phi r^\star\right\|, 
\end{align*}
from which, using Lemma \ref{projection_F_contraction_lemma}, we deduce that $
    \left\|\tilde{Q}-Q^\star\right\|\leq 2\left[(1-\gamma)\sqrt{(1-\gamma^2)}\right]^{-1}\left\|\tilde{Q}-\Phi r^\star\right\|$,
after which by \eqref{v_approx_q_approx_bound}, we finally obtain
\begin{align*}
        \mathbb{E}\left[v^{\boldsymbol{\hat{\pi}}}(z_0)\right]-\mathbb{E}\left[v^{\boldsymbol{\tilde{\pi}}}(z_0)\right]\leq  2\left[(1-\gamma)\sqrt{(1-\gamma^2)}\right]^{-1}\left\|\tilde{Q}-\Phi r^\star\right\|,
\end{align*}
as required.
\end{proof}

Let us rewrite the update in the following way:
\begin{align*}
    r_{t+1}=r_t+\gamma_t\Xi(w_t,r_t),
\end{align*}
where the function $\Xi:\mathbb{R}^{2d}\times \mathbb{R}^p\to\mathbb{R}^p$ is given by:
\begin{align*}
\Xi(w,r):=\phi(z)\left(\Theta(z)+\gamma\max\left\{(\Phi r) (z'),\mathcal{M}(\Phi r) (z')\right\}-(\Phi r)(z)\right),
\end{align*}
for any $w\equiv (z,z')\in\left(\mathbb{N}\times\mathcal{S}\right)^2$ where $z=(t,s)\in\mathbb{N}\times\mathcal{S}$ and $z'=(t,s')\in\mathbb{N}\times\mathcal{S}$  and for any $r\in\mathbb{R}^p$. Let us also define the function $\boldsymbol{\Xi}:\mathbb{R}^p\to\mathbb{R}^p$ by the following:
\begin{align*}
    \boldsymbol{\Xi}(r):=\mathbb{E}_{w_0\sim (\mathbb{P},\mathbb{P})}\left[\Xi(w_0,r)\right]; w_0:=(z_0,z_1).
\end{align*}
\begin{lemma}\label{iteratation_property_lemma}
The following statements hold for all $z\in \{0,1\}\times \mathcal{S}$:
\begin{itemize}
    \item[i)] $
(r-r^\star)\boldsymbol{\Xi}_k(r)<0,\qquad \forall r\neq r^\star,    
$
\item[ii)] $
\boldsymbol{\Xi}_k(r^\star)=0$.
\end{itemize}
\end{lemma}
\begin{proof}
To prove the statement, we first note that each component of $\boldsymbol{\Xi}_k(r)$ admits a representation as an inner product, indeed: 
\begin{align*}
\boldsymbol{\Xi}_k(r)&=\mathbb{E}\left[\phi_k(z_0)(\Theta(z_0)+\gamma\max\left\{\Phi r(z_1),\mathcal{M}\Phi(z_1)\right\}-(\Phi r)(z_0)\right] 
\\&=\mathbb{E}\left[\phi_k(z_0)(\Theta(z_0)+\gamma\mathbb{E}\left[\max\left\{\Phi r(z_1),\mathcal{M}\Phi(z_1)\right\}|z_0\right]-(\Phi r)(z_0)\right]
\\&=\mathbb{E}\left[\phi_k(z_0)(\Theta(z_0)+\gamma P\max\left\{\left(\Phi r,\mathcal{M}\Phi\right)\right\}(z_0)-(\Phi r)(z_0)\right]
\\&=\left\langle\phi_k,\mathfrak{F}\Phi r-\Phi r\right\rangle,
\end{align*}
using the iterated law of expectations and the definitions of $P$ and $\mathfrak{F}$.

We now are in position to prove i). Indeed, we now observe the following:
\begin{align*}
\left(r-r^\star\right)\boldsymbol{\Xi}_k(r)&=\sum_{l=1}\left(r(l)-r^\star(l)\right)\left\langle\phi_l,\mathfrak{F}\Phi r -\Phi r\right\rangle
\\&=\left\langle\Phi r -\Phi r^\star, \mathfrak{F}\Phi r -\Phi r\right\rangle
\\&=\left\langle\Phi r -\Phi r^\star, (\boldsymbol{1}-\Pi)\mathfrak{F}\Phi r+\Pi \mathfrak{F}\Phi r -\Phi r\right\rangle
\\&=\left\langle\Phi r -\Phi r^\star, \Pi \mathfrak{F}\Phi r -\Phi r\right\rangle,
\end{align*}
where in the last step we used the orthogonality of $(\boldsymbol{1}-\Pi)$. We now recall that $\Pi \mathfrak{F}\Phi r^\star=\Phi r^\star$ since $\Phi r^\star$ is a fixed point of $\Pi \mathfrak{F}$. Additionally, using Lemma \ref{projection_F_contraction_lemma} we observe that $\|\Pi \mathfrak{F}\Phi r -\Phi r^\star\| \leq \gamma \|\Phi r -\Phi r^\star\|$. With this we now find that
\begin{align*}
&\left\langle\Phi r -\Phi r^\star, \Pi \mathfrak{F}\Phi r -\Phi r\right\rangle    
\\&=\left\langle\Phi r -\Phi r^\star, (\Pi \mathfrak{F}\Phi r -\Phi r^\star)+ \Phi r^\star -\Phi r\right\rangle
\\&\leq\left\|\Phi r -\Phi r^\star\right\|\left\|\Pi \mathfrak{F}\Phi r -\Phi r^\star\right\|- \left\|\Phi r^\star -\Phi r\right\|^2
\\&\leq(\gamma -1)\left\|\Phi r^\star -\Phi r\right\|^2,
\end{align*}
which is negative since $\gamma<1$ which completes the proof of part i).

The proof of part ii) is straightforward since we readily observe that
\begin{align*}
    \boldsymbol{\Xi}_k(r^\star)= \left\langle\phi_l, \mathfrak{F}\Phi r^\star-\Phi r\right\rangle= \left\langle\phi_l, \Pi \mathfrak{F}\Phi r^\star-\Phi r\right\rangle=0,
\end{align*}
as required and from which we deduce the result.
\end{proof}
To prove the theorem, we make use of a special case of the following result:

\begin{theorem}[Th. 17, p. 239 in \cite{benveniste2012adaptive}] \label{theorem:stoch.approx.}
Consider a stochastic process $r_t:\mathbb{R}\times\{\infty\}\times\Omega\to\mathbb{R}^k$ which takes an initial value $r_0$ and evolves according to the following:
\begin{align}
    r_{t+1}=r_t+\alpha \Xi(s_t,r_t),
\end{align}
for some function $s:\mathbb{R}^{2d}\times\mathbb{R}^k\to\mathbb{R}^k$ and where the following statements hold:
\begin{enumerate}
    \item $\{s_t|t=0,1,\ldots\}$ is a stationary, ergodic Markov process taking values in $\mathbb{R}^{2d}$
    \item For any positive scalar $q$, there exists a scalar $\mu_q$ such that $\mathbb{E}\left[1+\|s_t\|^q|s\equiv s_0\right]\leq \mu_q\left(1+\|s\|^q\right)$
    \item The step size sequence satisfies the Robbins-Monro conditions, that is $\sum_{t=0}^\infty\alpha_t=\infty$ and $\sum_{t=0}^\infty\alpha^2_t<\infty$
    \item There exists scalars $c$ and $q$ such that $    \|\Xi(w,r)\|
        \leq c\left(1+\|w\|^q\right)(1+\|r\|)$
    \item There exists scalars $c$ and $q$ such that $
        \sum_{t=0}^\infty\left\|\mathbb{E}\left[\Xi(w_t,r)|z_0\equiv z\right]-\mathbb{E}\left[\Xi(w_0,r)\right]\right\|
        \leq c\left(1+\|w\|^q\right)(1+\|r\|)$
    \item There exists a scalar $c>0$ such that $
        \left\|\mathbb{E}[\Xi(w_0,r)]-\mathbb{E}[\Xi(w_0,\bar{r})]\right\|\leq c\|r-\bar{r}\| $
    \item There exists scalars $c>0$ and $q>0$ such that $
        \sum_{t=0}^\infty\left\|\mathbb{E}\left[\Xi(w_t,r)|w_0\equiv w\right]-\mathbb{E}\left[\Xi(w_0,\bar{r})\right]\right\|
        \leq c\|r-\bar{r}\|\left(1+\|w\|^q\right) $
    \item There exists some $r^\star\in\mathbb{R}^k$ such that $\boldsymbol{\Xi}(r)(r-r^\star)<0$ for all $r \neq r^\star$ and $\bar{s}(r^\star)=0$. 
\end{enumerate}
Then $r_t$ converges to $r^\star$ almost surely.
\end{theorem}

In order to apply the Theorem \ref{theorem:stoch.approx.}, we show that conditions 1 - 7 are satisfied.

\begin{proof}
Conditions 1-2 are true by assumption while condition 3 can be made true by choice of the learning rates. Therefore it remains to verify conditions 4-7 are met.   

To prove 4, we observe that
\begin{align*}
\left\|\Xi(w,r)\right\|
&=\left\|\phi(z)\left(\Theta(z)+\gamma\max\left\{(\Phi r) (z'),\mathcal{M}\Phi (z')\right\}-(\Phi r)(z)\right)\right\|
\\&\leq\left\|\phi(z)\right\|\left\|\Theta(z)+\gamma\left(\left\|\phi(z')\right\|\|r\|+\mathcal{M}\Phi (z')\right)\right\|+\left\|\phi(z)\right\|\|r\|
\\&\leq\left\|\phi(z)\right\|\left(\|\Theta(z)\|+\gamma\|\mathcal{M}\Phi (z')\|\right)+\left\|\phi(z)\right\|\left(\gamma\left\|\phi(z')\right\|+\left\|\phi(z)\right\|\right)\|r\|.
\end{align*}
Now using the definition of $\mathcal{M}$, we readily observe that $\|\mathcal{M}\Phi (z')\|\leq \| \Theta\|+\gamma\|\mathcal{P}^\pi_{s's_t}\Phi\|\leq \| \Theta\|+\gamma\|\Phi\|$ using the non-expansiveness of $P$.

Hence, we lastly deduce that
\begin{align*}
\left\|\Xi(w,r)\right\|
&\leq\left\|\phi(z)\right\|\left(\|\Theta(z)\|+\gamma\|\mathcal{M}\Phi (z')\|\right)+\left\|\phi(z)\right\|\left(\gamma\left\|\phi(z')\right\|+\left\|\phi(z)\right\|\right)\|r\|
\\&\leq\left\|\phi(z)\right\|\left(\|\Theta(z)\|+\gamma\| \Theta\|+\gamma\|\psi\|\right)+\left\|\phi(z)\right\|\left(\gamma\left\|\phi(z')\right\|+\left\|\phi(z)\right\|\right)\|r\|,
\end{align*}
we then easily deduce the result using the boundedness of $\phi,\Theta$ and $\psi$.

Now we observe the following Lipschitz condition on $\Xi$:
\begin{align*}
&\left\|\Xi(w,r)-\Xi(w,\bar{r})\right\|
\\&=\left\|\phi(z)\left(\gamma\max\left\{(\Phi r)(z'),\mathcal{M}\Phi(z')\right\}-\gamma\max\left\{(\Phi \bar{r})(z'),\mathcal{M}\Phi(z')\right\}\right)-\left((\Phi r)(z)-\Phi\bar{r}(z)\right)\right\|
\\&\leq\gamma\left\|\phi(z)\right\|\left\|\max\left\{\phi'(z') r,\mathcal{M}\Phi'(z')\right\}-\max\left\{(\phi'(z') \bar{r}),\mathcal{M}\Phi'(z')\right\}\right\|+\left\|\phi(z)\right\|\left\|\phi'(z) r-\phi(z)\bar{r}\right\|
\\&\leq\gamma\left\|\phi(z)\right\|\left\|\phi'(z') r-\phi'(z') \bar{r}\right\|+\left\|\phi(z)\right\|\left\|\phi'(z) r-\phi'(z)\bar{r}\right\|
\\&\leq \left\|\phi(z)\right\|\left(\left\|\phi(z)\right\|+ \gamma\left\|\phi(z)\right\|\left\|\phi'(z') -\phi'(z') \right\|\right)\left\|r-\bar{r}\right\|
\\&\leq c\left\|r-\bar{r}\right\|,
\end{align*}
using Cauchy-Schwarz inequality and  that for any scalars $a,b,c$ we have that $
    \left|\max\{a,b\}-\max\{b,c\}\right|\leq \left|a-c\right|$.
    
Using Assumptions 3 and 4, we therefore deduce that
\begin{align}
\sum_{t=0}^\infty\left\|\mathbb{E}\left[\Xi(w,r)-\Xi(w,\bar{r})|w_0=w\right]-\mathbb{E}\left[\Xi(w_0,r)-\Xi(w_0,\bar{r})\right\|\right]\leq c\left\|r-\bar{r}\right\|(1+\left\|w\right\|^l).
\end{align}

Part 2 is assured by Lemma \ref{projection_F_contraction_lemma} while Part 4 is assured by Lemma \ref{value_difference_Q_difference} and lastly Part 8 is assured by Lemma \ref{iteratation_property_lemma}.
\end{proof}

\clearpage

\end{document}